\documentclass{article}
\usepackage{fullpage}
\usepackage{booktabs} 
\usepackage[ruled,noend]{algorithm2e} 

\SetAlFnt{\small}
\SetAlCapFnt{\small}
\SetAlCapNameFnt{\small}
\SetAlCapHSkip{0pt}
\IncMargin{-\parindent}

\usepackage{todonotes}
	\presetkeys{todonotes}{inline}{}
\usepackage{mathtools}
\DeclarePairedDelimiter{\ceil}{\lceil}{\rceil}
\DeclarePairedDelimiter{\floor}{\lfloor}{\rfloor}

\usepackage{amsmath,amssymb,amsthm}
\usepackage{thmtools} 
\usepackage{thm-restate}
\usepackage{cleveref}
\usepackage{tikz}
\usepackage{complexity}

\usetikzlibrary{positioning,arrows.meta,shapes,snakes,calc,fit,decorations.markings}
\usetikzlibrary{tikzmark}

\newtheorem{theorem}{Theorem}[section]

\newtheorem{lemma}[theorem]{Lemma}

\theoremstyle{definition}
\newtheorem{definition}[theorem]{Definition}

\newtheorem{specification}{Specification}

\newcommand{\FavorIn}{F^{\text{in}}} 
\newcommand{\FavorOut}{F^{\text{out}}} 
 
\newcommand{\OutEvent}[1][C,a]{\FavorOut(#1) = \emptyset} 
\newcommand{\PosEvent}[1][a,b]{\text{Tol}_{#1}} 
\newcommand{\FailEvent}[1][F]{\text{Fail}_{#1}}

\newcommand{\exv}{exit-denying} 
\newcommand{\env}{entry-denying} 
\newcommand{\EXVNoun}{Exit Denial} 
\newcommand{\ExvNoun}{Exit denial} 
\newcommand{\ENVNoun}{Entry Denial} 
\newcommand{\EnvNoun}{Entry Denial} 
\newcommand{\exvNoun}{exit denial} 
\newcommand{\envNoun}{entry denial} 
\newcommand{\partition}{\pi} 
\newcommand{\agone}{a} 
\newcommand{\agtwo}{b} 
\newcommand{\partnumber}{20}
\newcommand{\ind}{\iota} 
\newcommand{\coa}{\mathcal C} 
\newcommand{\mcoa}{\mathcal M} 
\newcommand{\Prob}{\mathbb{P}}
\newcommand{\alg}{\text{ALG}} 
\newcommand{\distr}{\mathcal D} 
 
\newcommand{\randgame}[1][\distr]{H(n,#1)} 
\newcommand{\Stwo}[1][N,u]{\mathcal S(#1)} 

\newcommand{\Ex}{\mathbb{E}}
\newcommand{\ER}{Erd\H{o}s-R{\'e}nyi}

\newcommand{\gdycliques}{\textsc{GreedyCliqueFormation}} 
\newcommand{\gdycluster}{\textsc{GreedyClustering}} 
\DeclareRobustCommand{\stirling}{\genfrac\{\}{0pt}{}}

\usepackage[round]{natbib}
\usepackage{authblk} 
\title{Stability in Random Hedonic Games}

\author{Martin Bullinger}
\author{Sonja Kraiczy}
\affil{ \small Department of Computer Science, University of Oxford, UK\protect\\ \vspace*{0.1cm} martin.bullinger@cs.ox.ac.uk, sonja.kraiczy@merton.ox.ac.uk}
\date{}

\allowdisplaybreaks 

\begin{document}

\maketitle

\begin{abstract}
Partitioning a large group of employees into teams can prove difficult because unsatisfied employees may want to transfer to other teams.
In this case, the team (coalition) formation is unstable and incentivizes deviation from the proposed structure. 
Such a coalition formation scenario can be modeled in the framework of hedonic games
and a significant amount of research has been devoted to the study of stability in such games.
Unfortunately, stable coalition structures are not guaranteed to exist in general 
and their practicality is further hindered by computational hardness barriers.
We offer a new perspective on
this matter by studying a random model of hedonic games. 
For three prominent stability concepts based on single-agent deviations, we provide a high probability analysis of stability in the large agent limit.

Our first main result is an efficient algorithm that outputs an individually and contractually Nash-stable partition with high probability. 
Our second main result is that the probability that a random game admits a Nash-stable partition tends to zero.
Our approach resolves the two major downsides associated with individual stability and contractual Nash stability and reveals
agents acting single-handedly are usually to blame for instabilities.
\end{abstract}

\section{Introduction}
Consider the economic problem of partitioning employees into teams which work on distinct tasks. 
In such a setting, the composition of the team may impact the individual team members productivity and well-being. As a result, an employee may want to change team if they find that a different team would be a better match.
We are therefore faced with the challenge of finding a team structure 
that reflects the employees' preferences so as to prevent such instabilities. 

The above scenario is an example of a more general setting called \emph{coalition formation}, in which the goal is to partition a set of agents into so-called \emph{coalitions}.
We consider coalition formation in the framework of \emph{hedonic games}, where agents compare their own coalition to other coalitions they could join. \citep{DrGr80a}. 
Externalities such as the wider coalition structure are not taken into account.

For a deviation to another coalition to be reasonable, the deviating agent should prefer their new coalition to the abandoned coalition.
The absence of such deviations yields the classical concept of 
\emph{Nash stability} \citep{BoJa02a}. 
A standard example, sometimes referred to as \emph{run-and-chase}, illustrates that Nash stability is demanding even with just two agents.
While Alice prefers to be in a coalition on her own, Bob prefers to be in a coalition with Alice.
As a result, neither the coalition structure in which both agents are on their own nor the coalition structure in which they form one joint coalition is Nash-stable.

So far, we have entirely disregarded the preference of the abandoned team and the new team. 
However, in a real-world setting an employee may be more likely to deviate if either their current team or their preferred team also benefit from the proposed change.
In light of this, research in hedonic games has considered stability concepts of varying flexibility. Most prominent among them are variations of the concept of Nash deviation that require that the agents in the abandoned or in the joined coalition do not fare worse after the deviation---resulting in \textit{contractual Nash stability} and \textit{individual stability}, respectively \citep{BoJa02a,SuDi07b}. 
Despite modeling more realistic deviations, these stability concepts remain unsatisfactory for two reasons. 
Firstly, there exist simple situations in which they cannot be satisfied, so-called \textit{No-instances} [\citealp[Example~5]{BoJa02a}, \citealp[Example~2]{SuDi07b}].
Secondly, deciding whether a given hedonic game admits a stable coalition structure is usually computationally hard. 
There is an abundance of research concerning these and other stability concepts proving analogous results \citep[see, e.g.,][]{SuDi10a,ABS11c,Woeg13a,BBW21b,BBT22a,Bull22a}.\footnote{\citet[Chapter~4.3]{Bull23a} discusses a general method (applied by most of the cited literature) for the usage of No-instances to obtain hardness results.}
Many of the No-instances seem unrealistic and delicate in their design and therefore just within the combinatorial ``reach'' of the representation model.\footnote{For some solution concepts, even the smallest known counterexamples require a large number of agents \citep[see, e.g.,][]{ABB+17a,Bull22a} and some instances were only found by computer-aided search \citep[Theorem~5.1]{BBW21b}.} Motivated by these observations, we aim to answer the following natural question.
\begin{center}
    \textit{Do hedonic games typically admit stable coalition structures?}
\end{center}

\paragraph{Our Contribution}{In this work, we complement the worst-case analysis of stability in hedonic games with a \textit{high probability analysis}. 
We define a random game model based on the prominent class of additively separable hedonic games \citep{BoJa02a}.
In these games, agents' preferences are encoded by cardinal utility values for other agents.
The utility for a coalition then simply is the sum of the utility values for the agents in this coalition.
An agent aims at maximizing the utility of their coalition within a coalition structure.
We assume that, for every pair of agents $\agone$ and $\agtwo$, the utility value of $\agone$ for $\agtwo$ is independently and identically distributed according to some given distribution.
In this model, we analyse three stability notions described above: Nash stability, contractual Nash stability and individual stability.
We determine the limit behavior of the probability that a stable outcome exists in a random game if the number of agents tends to infinity---the so-called large agent limit.}

We start with two results that give insight on when simple coalition structures are stable.
First, whenever the utility distribution is positive with positive probability, then the coalition structure, where all agents form a single coalition, the so-called \emph{grand coalition}, is contractually Nash-stable.
This follows because for additive utility aggregation, a positive utility denies another agent to leave.
Hence,
in the large agent limit, the probability that an agent is denied to leave is high.
If, however, the distribution only has nonpositive support, then all utilities are negative and the coalition where every agent forms a coalition on their own is Nash-stable.
Second, if there is a bias towards a positive single-agent utility, then, in the large agent limit, the accumulated utility is positive for all agents, and consequently, the grand coalition is Nash-stable.

\begin{figure}
    \centering
    \begin{tikzpicture}
        \node (EXV) at (-3.5,3) {\textbf{\EXVNoun}};
        \node (ENV) at (5,3) {\textbf{{\ENVNoun} $\land$ Individual Rationality}};
        \node (NS) at (0,3) {Nash Stability};
        \node (IS) at (2.5,1.5) {Individual Stability};
        \node (CNS) at (-2.5,1.5) {Contractual Nash Stability};
        \node (IR) at (4.5,0) {Individual Rationality};
        \node (CIS) at (0,0) {Contractual Individual Stability};
        
        \draw[->] 
        (NS) edge (IS)
        (NS) edge (CNS)
        (ENV) edge (IS)
        (EXV) edge (CNS)
        (IS) edge (IR)
        (CNS) edge (CIS)
        (IS) edge (CIS);

    \end{tikzpicture}
    \caption{Stability concepts for hedonic games, where arrows represent logical relationships; see \Cref{sec:stab} for formal definitions.
    Contractual individual stability and individual rationality can always be satisfied while there exist No-instances for all other solution concepts. 
    }
    \label{fig:concepts}
\end{figure}
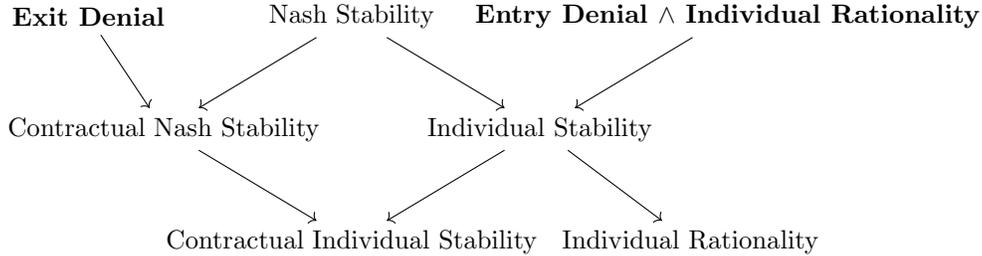

In the remaining paper, we perform an extensive analysis of our random game model if utilities are sampled from uniform distributions with mean~$0$. 
These distributions already allow for rich utility profiles and their symmetry around $0$ results in complex behavior.
Our main result is an efficient algorithm that, with high probability, produces outputs that are
\begin{itemize}
    \item individually rational: every agent has a nonnegative utility,
    \item {\env}: every agent is denied to join any other coalition by some agent in the other coalition, and
    \item {\exv}: every agent is denied to abandon their own coalition by some agent in their own coalition.
\end{itemize}

{\EXVNoun} is strictly stronger than contractual Nash stability.
In particular, it can never hold for coalition structures containing singleton coalitions, because such coalitions do not contain other agents to deny exit. Moreover, {\envNoun} does only imply individual stability in conjunction with individual rationality because {\envNoun} only prevents individual deviations to nonempty coalitions.
The logical relationships of all our solution concepts are depicted in \Cref{fig:concepts}.

Our algorithm consists of three stages: in the first stage, we apply a clustering coalitions to find coalitions in which all agents have mutually high utility.
In the second stage, we merge the coalitions from the first stage to larger coalitions.
In these two stages, we form coalitions of most but not all agents.
Hence, in a third stage, we greedily add the remaining agents to the coalitions created in the first two stages.
Intuitively, the output of this algorithm is individually rational because of the high utility achieved in the first stage, whereas the merging stage increases the size of coalitions to obtain coalitions that are large enough to deny agents to entry or exit.
With high probability, all coalitions are of size $\Theta(\log n)$ and almost equal in size, giving a much more practical coalition structure than the grand or singleton coalitions.

In contrast to this positive algorithmic result, we prove that with high probability, no partition is Nash-stable if we sample utilities from $U(-1,1)$ (or equivalently any other mean $0$ uniform distribution).
Our result is based on a sophisticated argument which relies on counting and estimation.

\section{Related Work}

Coalition formation has been an important concern in economic theory and can be traced back as far as to the origins of game theory \citep{vNM44a}.
The framework of hedonic games was introduced by \citet{DrGr80a} and further formalized by \citet{BoJa02a}, \citet{BKS01a}, and \citet{CeRo01a}.
Since these publications, hedonic games have received a constant stream of attention;
see the survey by \citet{AzSa15a} for an overview.

The general model of hedonic games poses two challenges: 
agents' preferences have to consider an exponentially large set of coalitions they can be part of and solution concepts, especially notions of stability, can often not be satisfied.
The main objective of hedonic games research is therefore often to identify succinctly representable classes of games, in which desirable solution concepts can be satisfied.

This has led to a rich variety of game classes.
While we investigate additively separable hedonic games (ASHGs), in which cardinal valuations are aggregated by sums \citep{BoJa02a}, other game classes are based on aggregating cardinal valuations by taking averages \citep{ABB+17a,Olse12a}, ranking agents \citep{CeRo01a,CeHa04a}, specifying friends and enemies \citep{DBHS06a}, or encoding preferences by Boolean formulae \citep{ElWo09a}.

While stability and other solution concepts have been considered in all of these classes, we focus on results for ASHGs.
Two properties have been identified to be vital for achieving stability based on deviations by single agents.
First, symmetric utilities allow the existence of Nash-stable (and therefore individually stable and contractually Nash-stable) outcomes \citep{BoJa02a}.
This follows from a local search algorithm that also leads to Nash stability in related settings, such as swap stability in ASHGs with fixed-size coalitions \citep{BMM22a} or for Nash stability in a generization of ASHGs \citep{BuSu24a}.
However, in accordance with the nature of this approach, computing Nash-stable and even individually stable partitions in symmetric ASHGs is \PLS-complete \citep{GaSa19a}.
Second, if utilities are not assumed to be symmetric, individual stability and contractual Nash stability can be obtained if utilities are essentially restricted to one positive and one negative value \citep{BBT22a}.
Importantly, this captures subclasses of ASHGs based on the distinction of friends and enemies \citep{DBHS06a}.

For the unrestricted utility model of ASHGs, deciding whether Nash-stable, individually stable, or contractually 
Nash-stable outcomes exists, is \NP-complete \citep{SuDi10a,Bull22a}.
By contrast, one can consider the case where deviations need the unanimous consent of both the abandoned and the joined coalition, which defines the concept of contractual individual stability. 
Then, stable outcomes are guaranteed to exist (because every deviation increases the social welfare) and can be computed in polynomial time \citep{ABS11c}.
An interesting contrast to this is the work by \citet{BuRo24a}, which considers stability in an online model of ASHGs and shows that contractual individual stability can only be obtained for restricted utility values.
For group stability, captured by the so-called core, stable outcomes once again exist for restricted utilities \citep{DBHS06a} while their computation for unrestricted ASHGs is $\Sigma_2^p$-complete
\citep{Woeg13a}.
ASHGs have also been studied for solution concepts beyond stability.
These encompass welfare optimality, Pareto optimality, fairness, or strategyproofness \citep{ABS11c,Bull19a,FKMZ21a,EFF20a}.

Also in hedonic games beyond ASHGs, existence of stable outcomes is rare, and usually can only be achieved under severe utility restrictions.
For instance, in fractional hedonic games, Nash-stable outcomes exist if weights are symmetric and restricted to be $0$ or $1$ \citep{BBS14a}, whereas symmetric but unrestricted utilities do not even guarantee individual stability \citep{BBW21b}, which is an interesting contrast to ASHGs.
A different approach was undertaken by \citet{AGIM24a} who consider axioms guaranteeing group stability in a general coalition formation context.
Due to the nature of this approach, their imposed axioms have to be quite strong, and their key structural requirements fail in our random games with high probability.

All of the preceding work considers deterministic games and hardness results can be viewed as a worst-case analysis.
By contrast, we propose the study of random hedonic games.
While this is the first such study for hedonic games, similar approaches have been explored in other settings, especially in voting.
For instance, with high probability Condorcet winners do not exist \citep{Plot67a} while common tournament solutions such as the top cycle or the Banks set are equal to the entire set of alternatives \citep{Fey08a}.
In the domain of coalition formation, an interesting orthogonal approach to ours was undertaken by \citet{FFKV23a} who consider randomized solution concepts instead of random games.
They study two classes of coalition formation games different from ASHGs and provide algorithms computing outcomes that do not admit group deviations with high probability.
Interestingly, similar to our work, guaranteeing individual rationality seems to be a challenge in their setting and is not achieved by their algorithms.

\section{Preliminaries}

In this section, we introduce hedonic games and stability concepts. 
We assume familiarity with basic concepts from graph theory, such as graphs, cliques, or colorings and refer the reader to the text book by \citet{Dies05a} for more detailed background.
We use the notation $[k] := \{1,\dots, k\}$ for any positive integer $k$.

\subsection{Hedonic Games}\label{sec:HG}
Throughout the paper, we consider settings with a set $N$ of $n$ agents.
The goal of coalition formation is to find a partition of the agents into disjoint coalitions according to their preferences. 
A \emph{coalition structure} (or \emph{partition}) of $N$ is a subset $\partition\subseteq 2^N$ such that $\bigcup_{C\in \partition} C = N$, and for every pair $C,D\in \partition$, it holds that $C = D$ or $C\cap D = \emptyset$. 
A nonempty subset of $N$ is called a \emph{coalition}.
Hence, every element of a partition is a coalition.
Given a partition~$\partition$, we denote by $\partition(\agone)$ the coalition containing agent~$\agone$.
We refer to the partition $\partition$ given by $\partition(\agone)=\{\agone\}$ for every agent $\agone\in N$ as the \emph{singleton partition}, and to $\partition=\{N\}$ as the \emph{grand coalition}.

We now introduce some terminology for partitions.
Given a partition $\partition$ and a positive integer $s\in \mathbb N$, $\partition$ is said to be \emph{$s$-balanced} if $|C| = s$ for all $C\in \partition$.
In addition, $\partition$ is said to be $s^+$-balanced if $|C| = s$  or $|C| = s+1$ for all $C\in \partition$. 
This notion is important for settings where $s$ does not divide the number of agents. 
Given an agent set $N$, a partition $\partition'$ of a subset $N'\subseteq N$ is called a \emph{partial partition}.
Given a partial partition $\partition'$, we denote $N(\partition') := \bigcup_{C\in \partition'}C$, i.e., $N(\partition)$ is the agent set partitioned by $\partition'$.
A sequence $(\partition_1,\dots, \partition_k)$ of partial partitions is said to be \emph{disjoint} if for all $1\le i < j\le k$ it holds that $N(\partition_i) \cap N(\partition_j) = \emptyset$.

We are ready to define hedonic games, the central concept of this paper.
Let $\mathcal N_{\agone}$ denote all possible coalitions containing agent $\agone$, i.e., $\mathcal N_{\agone}=\{C\subseteq N\colon \agone\in C\}$.
A \emph{hedonic game} is defined by a pair $(N,\succsim)$, where $N$ is a finite set of agents and ${\succsim} = (\succsim_{\agone})_{\agone\in N}$ is a tuple of weak orders $\succsim_{\agone}$ over $\mathcal N_{\agone}$ which represent the preferences of the respective agent $\agone$. Hence, agents express preferences only over the coalitions which they are part of without considering externalities.
The strict part of an order $\succsim_{\agone}$ is denoted by $\succ_{\agone}$, i.e., $C\succ_{\agone} D$ if and only if $C \succsim_{\agone} D$ and not $D \succsim_{\agone} C$.

Within the general framework of hedonic games, a large variety of classes of games have been proposed in the literature.
Many of these rely on cardinal utility functions that assign a value for every agent.
We consider the prominent class of games, where this values are aggregated additively to values of coalitions. 
Following \citet{BoJa02a}, an \emph{additively separable hedonic games} is defined by a pair $(N,u)$ where $N$ is a finite set of agents and $u = (u_{\agone})_{\agone\in N}$ is the tuple of \emph{utility functions}.
For each $\agone\in N$, we have $u_{\agone}\colon N\to \mathbb Q$.
This induces the hedonic game $(N,\succsim)$ where for each $C,D\in \mathcal N_{\agone}$ we have $C\succsim_{\agone} D$ if and only if $\sum_{\agtwo\in C\setminus \{\agone\}} u_{\agone}(\agtwo) \ge \sum_{\agtwo\in D\setminus \{\agone\}} u_{\agone}(\agtwo)$.
We also represent the induced hedonic game by the pair $(N,u)$.
We extend the utility functions to coalitions $C \in \mathcal N_{\agone}$ or partitions $\partition$ by defining $u_{\agone}(C) := \sum_{j\in C\setminus \{\agone\}} u_{\agone}(j)$ and $u_{\agone}(\partition) := u_{\agone}(\partition(\agone))$.
Note that we view the empty sum as having a value of~$0$, and therefore a singleton coalition has a utility of~$0$.
Because of the utility functions for single agents, additively separable hedonic games can equivalently be represented as weighted (and complete) directed graphs.

In our algorithms, we are sometimes only interested in a subgame of a hedonic game induced by a subset of the agents.
Given a utility function $u_{\agone}$ and a subset of agents $M\subseteq N$, we denote by $u_{\agone}|_M$ the function $u_{\agone}|_M \colon M\to \mathbb Q$ defined by $\agtwo\mapsto u_{\agone}(\agtwo)$.
Also, we denote by $u|_M$ the tuple $u|_M = (u_{\agone}|_M)_{\agone\in N}$.
Given a cardinal hedonic game $(N,u)$ and a subset of players $M\subseteq N$, we call the cardinal hedonic game $(M,u|_M)$ the \emph{subgame induced by} $M$.

\subsection{Stability Based on Single-Agent Deviations}\label{sec:stab}

We now formally introduce stability.
We remind the reader that their logical relationship can be found in \Cref{fig:concepts}.
As mentioned before, we focus on stability notions based on deviations by single agents.
A \emph{single-agent deviation} performed by agent $\agone$ transforms a partition~$\partition$ into a partition $\partition'$ 
such that $\partition(\agone)\neq\partition'(\agone)$ and, for all agents $\agtwo\neq \agone$, it holds that $\partition(\agtwo)\setminus\{\agone\} = \partition'(\agtwo)\setminus\{\agone\}$.
We write $\partition \xrightarrow{\agone} \partition'$ to denote a single-agent deviation performed by agent $\agone$ transforming partition $\partition$ to partition $\partition'$.

We consider myopic agents whose rationale is to only engage in a deviation if it immediately makes them better off. 
A \emph{Nash deviation} is a single-agent deviation performed by agent $\agone$ making her better off, i.e., $\partition'(\agone)\succ_{\agone} \partition(\agone)$.
Any partition in which no Nash deviation is possible is said to be \emph{Nash-stable} (NS).

Like we have argued in the introduction, this concept of stability is very strong and comes with the drawback that only the preferences of the deviating agent are considered. Therefore, we also consider two natural refinements, which additionally require the consent of the abandoned or the joined coalition.
For a compact representation, we introduce them via the notion of favor sets. 
Let $C \subseteq N$ be a coalition and $\agone \in N$ an agent. The \emph{favor-in set} of $C$ with respect to $\agone$ is the set of agents in $C$ (excluding $\agone$) that strictly favor having $\agone$ inside $C$ rather than outside, i.e.,
$\FavorIn(C, \agone) = \{\agtwo \in C \setminus \{\agone\} \colon C \cup \{\agone\} \succ_{\agtwo} C \setminus \{\agone\}\}$. 
The \emph{favor-out set} of $C$ with respect to $\agone$ is the set of agents in $C$ (excluding $\agone$) that strictly favor having $\agone$ outside $C$ rather than inside, i.e., 
$\FavorOut(C, \agone) = \{\agtwo \in C \setminus \{\agone\} \colon C \setminus \{\agone\} \succ_{\agtwo} C \cup \{\agone\}\}$.

An \emph{individual deviation} is a Nash deviation $\partition \xrightarrow{\agone} \partition'$ such that $\FavorOut(\partition'(\agone), \agone) = \emptyset$.
A \emph{contractual deviation} is a Nash deviation $\partition \xrightarrow{\agone} \partition'$ such that $\FavorIn(\partition(\agone), \agone) = \emptyset$.
Then, a partition is said to be \emph{individually stable} (IS) or \emph{contractually Nash-stable} (CNS) if it allows for no individual or contractual deviation, respectively. A related weakening of both stability concepts is contractual individual stability (CIS), based on deviations that are both individual and contractual deviations~\citep{BoJa02a,SuDi07b}.

A weak form of stability is to only consider the case of deviations where agents form a singleton coalition.
A coalition $C\in \mathcal N_{\agone}$ of agent~$\agone$ is said to be \emph{individually rational} if $C\succsim_{\agone} \{\agone\}$, i.e., $C$ is at least as good as being in a singleton coalition.
Moreover, a partition $\partition$ is said to be \emph{individually rational} if, for every agent~$\agone$, the coalition $\partition(\agone)$ is individually rational.
Since the condition of consent by the joined coalition is trivial in this case, individual rationality is a weakening of individual stability.
However, in contrast to individual stability, individual rationality can always be satisfied by some partition since the singleton partition is individually rational.

Finally, we introduce two related concepts that are related to individual and contractual Nash stability and single out the denial of agents.
A partition $\partition$ is said to be \emph{\env} if for any agent $\agone \in N$ and $C\in \partition\setminus \{\partition(\agone)\}$ it holds that $\FavorOut(C,a)\neq \emptyset$.
Similarly, $\partition$ is said to be \emph{\exv} if for any agent $\agone \in N$ it holds that $\FavorIn(\partition(\agone),\agone)\neq \emptyset$.
Hence, this concepts capture the situation where joining or abandoning coalitions does not have consent.
It directly follows from the definition that {\exvNoun} implies contractual Nash stability and that {\exv} cannot hold for partitions containing singleton coalitions.
By contrast, {\envNoun} only implies individual stability in conjunction with individual rationality.
This is due to the fact that {\envNoun} does not constraint deviations to form a singleton coalition.

\subsection{Random Hedonic Games}

For our high probability analysis, we introduce a model of random hedonic games.
Let $n\in \mathbb{N}$ and $\distr$ be a probability distribution over $\mathbb R$. 
Let $\randgame$ be the distribution
over games $(N,u)$ where $N$ is a set of $n$ agents and for $\agone, \agtwo \in N$, $u_{\agone}(\agtwo)$ is independently sampled from $\mathcal D$.
Hence, $\randgame$ captures random games where utilities are independent and identically distributed (henceforth,~i.i.d.) samples from $\distr$.
We call a game $(N,u)$ sampled from $\randgame$ a \emph{random hedonic game}.
As an important special case, we consider the case where $\distr$ is the uniform distribution on the interval $(-1,1)$, denoted by $U(-1,1)$.\footnote{By scaling utilities, this is equivalent to sampling utilities from $U(-x,x)$ for any $x>0$.}
We are interested in determining the value of limits of the type
	\begin{align*}
		\lim_{n\to \infty} \Prob\left((N,u) \text{ contains Nash-stable partition}\right)\text.
	\end{align*}

There, Nash stability may be replaced by other properties, such as individual rationality or individual stability.
In all probability expressions, we assume that games are drawn randomly, i.e., $(N,u)\sim \randgame$, where $n = |N|$.
Since these expressions consider the limit behavior when the number of agents tends to infinity, we also speak of the \emph{large agent limit}.

Because of the representation of cardinal hedonic games by (complete and weighted directed graphs), we will apply theory for random graphs from probabilistic combinatorics. 
We denote by $G(n,p)$ the {\ER} random (undirected) graph with $n$ vertices where every edge is present with probability $p$ \citep{ErRe59a}.

\section{Simple Stable Partitions}\label{sec:simplestable}
In this section, we consider two special partitions: the partition into singletons and the grand coalition. 
We show that under certain conditions on the utility distributions of our random games, these partitions are stable in the large agent limit.
In particular, contractual Nash stability is always achievable with high probability using one of these partitions. We will see that this no longer holds for individual stability or Nash stability in general.
For the special case where the utility distribution has strictly positive mean and finite variance, the grand coalition is Nash Stable, and we will exhibit this for uniform distributions.

Of course, both of these partitions are somewhat degenerate and are likely not a desired output of a real coalition formation scenario. 
Nevertheless, it is useful to see that for contractual Nash stability, we \textit{can} give a simple probabilistic argument for two reasons.
Firstly, because, as already mentioned, the simple approach fails for Nash stability and individual stability in general.
Thus, secondly, these results call for a more sophisticated approach even for the $U(-1,1)$ distribution case and lay the foundation for the involved proofs of our main results in the upcoming sections.

We first show that with high probability there exists a contractually Nash-stable partition.
More precisely, we will show that if we sample utilities from a distribution with positive mass on $(0,\infty)$, then the grand coalition is contractually Nash-stable in the large agent limit.
Otherwise, all sampled utilities are nonpositive and the singleton partition is Nash-stable.

\begin{restatable}{proposition}{CNSposssibility}
Let $\distr$ be any probability distribution over $\mathbb R$ and let $(N,u)\sim \randgame$ be a random hedonic game.
Then,
	\begin{align*}
		\lim_{n\to \infty} \Prob\left((N,u) \text{ admits a contractually Nash-stable partition}\right) = 1\text.
	\end{align*}
\end{restatable}

\begin{proof}
	Let $\distr$ be any probability distribution over $\mathbb R$ and let $(N,u)\sim \randgame$.
Let $X\sim \distr$ be a random variable distributed according to $\distr$ and define $\epsilon:=\Prob(X > 0)$. 
Note that $\epsilon \ge 0$.
We perform a case distinction dependent on whether $\epsilon=0$ or $\epsilon >0$.
First, if $\epsilon=0$, then all utilities between agents in $(N,u)$ are nonpositive.
Hence, the singleton coalition is Nash-stable and therefore in particular contractually Nash-stable.

Now consider the case that $\epsilon >0$.
Then, by assumption, for any pair of agents $a,b\in N$ it holds that $\Prob(u_b(a)\le 0) = 1-\epsilon$. 
Hence, since the utilities are drawn independently, it holds for every player $a\in N$ that
	\begin{equation}\label{eq:deny}
		\Prob(\exists b\in N\setminus\{a\}\colon u_b(a) > 0) = 1 - (1-\epsilon)^{n-1}\text.
	\end{equation}
	
	Consequently,
	\begin{align*}
		&\Prob(\{N\} \text{ contractually Nash-stable for }(N,u))  \le \Prob(\{N\} \text{ \exv})\\
		& = \prod_{a\in N}\Prob(\exists b\in N\setminus\{a\}\colon u_b(a) > 0) \overset{\text{\Cref{eq:deny}}}{=} \left(1 - (1-\epsilon)^{n-1}\right)^{n} \xrightarrow{n \to \infty} 1.
	\end{align*}
	
	In the equality, we use that for every two players $a,a'\in N$ the events $\{\exists b\in N\setminus\{a\}\colon u_b(a) > 0\}$ and $\{\exists b\in N\setminus\{a'\}\colon u_b(a') > 0\}$ are independent because they depend on a different set of utilities.
\end{proof}

Note that the grand coalition is typically not Nash-stable if the utility distribution has a negative mean $\mu$ and constant variance $\sigma^2$. The problem is that the probability that the grand coalition provides negative overall utility tends to a positive constant since by the Central Limit Theorem the utility for the grand coalition is approximately Gaussian with mean $n\mu$. Therefore, in the large agent limit, some agent has negative utility.  

However, following the same reasoning, we see that if the mean $\mu$ of the distribution is strictly positive (and the variance is finite), then the utility of the grand coalition is concentrated around $(n-1)\mu>0$ and we can apply concentration inequalities.
When the distribution has bounded support, this is easily captured by Hoeffding's inequality.
\begin{lemma}[\citealp{Hoef63a}]\label{thm:Hoeffding}
Let $X_1,\ldots, X_n$ be independent random variables such that, for all $i\in [n]$, it holds that $a_i\leq X_i \leq b_i$. 
Let $S_n=\sum_{i=1}^n X_i$.
Then, $$\Prob(S_n -\Ex(S_n)\geq t)\leq e^{-\frac{2t^2}{\sum_{i=1}^n(b_i-a_i)^2}}.$$
\end{lemma}

It follows that no agent has negative utility with high probability. 

\begin{restatable}{proposition}{NashBias}\label{prop:NashBias}
Consider $a< b$ with $\frac{a+b}{2}>0$. 
Let $\distr=U(a,b)$ 
and $(N,u)\sim \randgame$ be a random hedonic game. 
Then,
\begin{align*}
		\lim_{n\to \infty} \Prob\left((N,u) \text{ admits a Nash-stable partition}\right) = 1\text.
	\end{align*}
\end{restatable}

\begin{proof}
For $i\in [n-1]$, let $X_i\sim \distr$.
Then, for agent $\agone\in N$, it holds that $u_{\agone}(N) = \sum_{i=1}^{n-1} X_i$, which describes the utility of $\agone$ for the grand coalition. 
We have that $\Ex\left(\sum_{i=1}^{n-1} X_i\right)=n\frac{a+b}{2}$.
Hence, by \citeauthor{Hoef63a}'s inequality
\begin{align*}\Prob\left(\sum_{i=1}^{n-1} X_i\leq 0\right)&=\Prob\left(\sum_{i=1}^{n-1} -X_i\geq 0\right)
\\&=\Prob\left(\sum_{i=1}^{n-1} -X_i+n\frac{a+b}{2}\geq(n-1)\frac{a+b}{2}\right)\\
&\leq e^{-\frac{(n-1)^2(a+b)^2}{2\cdot(b-a)(n-1)}}=e^{-\frac{(n-1)(a+b)^2}{2\cdot(b-a)}}\text.\end{align*}
The inequality follows from \Cref{thm:Hoeffding} with i.i.d. random variables $-X_1,\ldots,-X_{n-1}$ distributed according to $U(-b,-a)$.
    By a union bound over all agents, we can bound the probability that any agent has negative utility for the grand coalition.
    Hence,
    \begin{align*}
    &\Prob\left((N,u) \text{ admits a Nash-stable partition}\right)\\
    &\geq \Prob\left(\{N\} \text{ is Nash-stable for } (N,u)\right)\geq 1- ne^{-\frac{(n-1)(a+b)^2}{2\cdot(b-a)}}\text,\end{align*}
    which tends to $1$ as $n$ tends to infinity, as required.
\end{proof}

\section{Design of our Three-Stage Clustering Algorithm}\label{sec:2stageclust}

From our initial consideration in the previous section, we have gained two insights.
{\ExvNoun} is very powerful in achieving stability.
Moreover, another condition that even yields Nash stability is a reliable uniform bias towards positive utilities.
In this section, we work towards achieving stable partitions for the unbiased distribution $U(-1,1)$.
We will derive our main technical contribution, namely an algorithm that produces nondegenerate partitions with strong stability guarantees including individual stability. 
With high probability, this algorithm produces a partition with $\Theta\left(\frac{n}{\log n}\right)$ coalitions of size $\Theta(\log n)$.
We start with an overview of our approach.
Then, we present our algorithm and structural properties of our algorithm in this section, and stability guarantees in the subsequent section.

\subsection{Outline of the Analysis}\label{sec:proofoutline}
An advantage of our proof of the existence of individually stable partitions is that it is constructive.
In particular, we give an efficient algorithm to compute individually stable partitions with high probability.
We start by describing the key steps of our algorithm.

The general idea of the algorithm is that we want to create coalitions where every agent has a high utility.
We achieve this by considering an {\ER} random graph with edges induced by pairs of agents with mutually high utility.
We then run a greedy clique formation algorithm, based on the greedy graph coloring algorithm by \citet{grimmett1975colouring}, to subsequently find cliques in this graph which we use as coalitions.
This algorithm follows the simple idea of iterating through the set of agents and adding them to a tentative clique whenever possible.

We want our algorithm to satisfy three properties with high probability:
\begin{enumerate}
    \item Individual rationality: Every agent has nonnegative utility.
    \item {\EnvNoun}: For every agent $\agone$ and coalition $C$ other than the agent's coalition, some agent in $C$ has negative utility for $\agone$.
    \item {\ExvNoun}: For every agent $\agone$, some agent in their own coalition has positive utility for $\agone$.
\end{enumerate}

Together, these properties imply individual stability as well as contractual Nash stability.
Since {\exvNoun} follows for similar reasons like {\envNoun}, we focus on giving intuition on how to achieve {\envNoun}.

Unfortunately, the naive algorithm of iteratively forming cliques and removing them from the graph as a coalition is not sufficient for our purpose.
Firstly, at some point of the algorithm, there will only be a small number of agents remaining and reasonably large cliques are impossible to form.
Hence, we need to figure out a way to deal with \emph{remaining agents}.
Secondly, we need to tailor the proportion of the number of created coalitions and the sizes of the created coalitions.
If we run the simple algorithm of iteratively forming cliques, we end up with a coalition structure that satisfies individual rationality.
However, we do not have a straightforward way to achieve {\envNoun} because the obtained coalitions are too small to have every coalition reject every agent outside of the coalition with high probability.
Surprisingly, the created coalitions in the basic approach are only too small by a constant factor.
Hence, our solution is to merge a finite number of coalitions.
We can do this in a way where utilities only decrease a bit. 
Therefore, we maintain individual rationality while establishing {\envNoun}.

\begin{figure}
    \centering
    \pgfmathsetmacro\xbredth{1}
    \pgfmathsetmacro\ybredth{1}
    \begin{tikzpicture}
        \foreach[count = \k] \i in {0,3*\xbredth,9*\xbredth}
        {
        \foreach[count = \l] \j in {10.5*\ybredth,7.5*\ybredth,3.6*\ybredth,0}
        {\node[draw, circle, minimum size = 5em] (C\k\l) at (\i,\j){};
        }
        \draw[rounded corners = 1em] ($(-.8,-1.6)+(\i,0)$) rectangle ($(.8+\i,-4)$);
        \node[draw, circle, minimum size = .4em] at ($(.4+\i,-3.5)$) {};
        \node[draw, circle, minimum size = .4em] at ($(.2+\i,-2.8)$) {};
        \node[draw, circle, minimum size = .4em] at ($(.4+\i,-2.1)$) {};
        \node[draw, circle, minimum size = .4em] at ($(-.4+\i,-3.3)$) {};
        \node[draw, circle, minimum size = .4em] at ($(-.4+\i,-2.3)$) {};
        } 

        \foreach \k in {1,2,3}
        {
        \foreach \l in {1,2,3,4}
        {
        \node[draw, circle, minimum size = .4em] (a\k\l) at ($(C\k\l) + (90:.4)$) {};
        \node[draw, circle, minimum size = .4em] (b\k\l) at ($(C\k\l) + (210:.4)$) {};
        \node[draw, circle, minimum size = .4em] (c\k\l) at ($(C\k\l) + (330:.4)$) {};
        \draw 
        (a\k\l) edge (b\k\l)
        (b\k\l) edge (c\k\l)
        (c\k\l) edge (a\k\l);
        \node (dots\l) at (barycentric cs:C2\l=1,C3\l=1) {\Large\textbf{\dots}};
        }
        \node[rotate = 90] at (barycentric cs:C\k3=1,C\k4=1) {\Large\textbf{\dots}};
        \node[rotate = 90] at (barycentric cs:C\k2=1,C\k3=1) {\Large\textbf{\dots}};
        }

        \node at ($(dots4) + (0,-2.8)$) {\Large\textbf{\dots}};
        
        \foreach \i/\k in {10.5*\ybredth/1,7.5*\ybredth/2,3.6*\ybredth/{t'}} 
        {
        \draw[rounded corners = 2em] ($(-1.5\xbredth,-1.2\xbredth) + (0,\i)$) rectangle ($(1.5*\xbredth+9*\xbredth,1.2\xbredth) + (0,\i)$);
        \node at ($(1.9*\xbredth+9*\xbredth,\i)$) {$C_{\k}$};
        }
        \draw[rounded corners = 2em] ($(-1.5\xbredth,-4.2\xbredth) + (0,0)$) rectangle ($(1.5*\xbredth+9*\xbredth,1.2\xbredth) + (0,0)$);
        \node at ($(1.9*\xbredth+9*\xbredth,-1.5\xbredth)$) {$R$};
        
        \foreach \i/\k in {0/1,3*\xbredth/2,9*\xbredth/20}
        {
        \draw[rounded corners = 2em] ($(-1.2\xbredth,-1.5\xbredth) + (\i,-3)$) rectangle ($(1.2\xbredth,1.5*\xbredth+10.4*\ybredth) + (\i,0)$);
        \node at ($(\i,1.9*\xbredth+10.3*\ybredth)$) {$N^{\k}$};
        }

        \foreach[count = \k] \i in {1,2,20}
        {
        \foreach[count = \l] \j in {1,2,{t'},t}
        {
        \node (l\k\l) at ($(C\k\l) + (45:1.2)$) {$C_{\j}^{\i}$};
        
        }
        \node at ($(l\k4)+(0.05,-2.3)$) {$R^{\i}$};
        }
    \end{tikzpicture}
    \caption{Visualization of our coalition formation algorithm.
    In the first stage, illustrated by the vertical sets, we consider $G^i$ for $i\in [\partnumber]$ and form clique coalitions $C^i_1,\dots,C^1_t$ by \Cref{alg:greedycol}.
    In each of the $G^i$, we have a set $R^i_1$ of remaining vertices.
    In the second stage, illustrated by horizontal coalitions, we merge the cliques obtained in the first stage by \Cref{alg:greedy} form the merged coalitions $C_1,\dots, C_{t'}$. 
    We are left with the remainder set $R$ of agents that we add to the obtained coalitions in the final stage.
    }
    \label{fig:algo}
\end{figure}
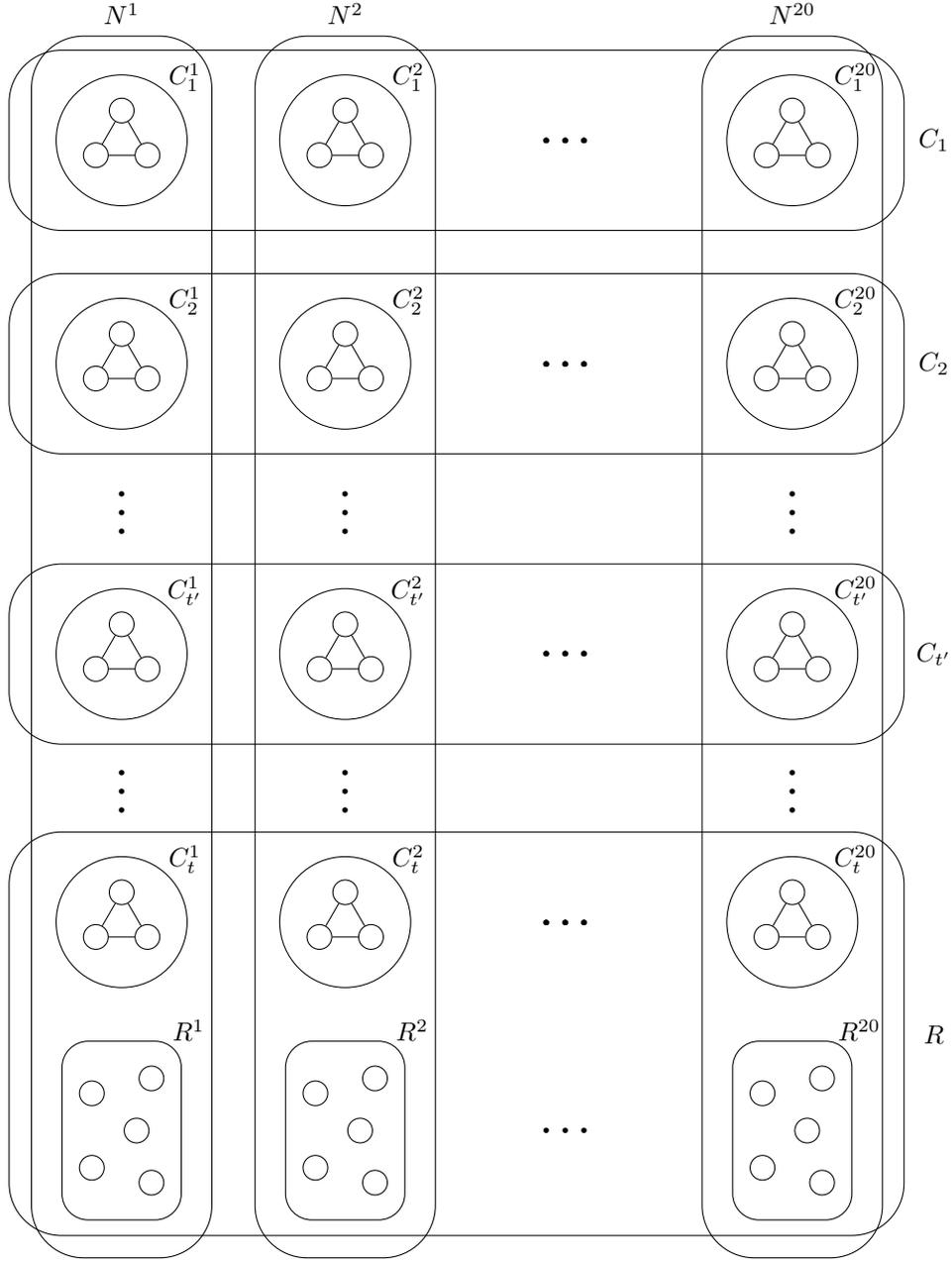

We now describe our main algorithm, whose basic structure is presented as \Cref{alg:main} and which is visualized in \Cref{fig:algo}.
First, we partition the agent set into $20$ subsets $N^1,\dots, N^{\partnumber}$ of (almost) equal size.
These sets are indicated by the vertical boxes in \Cref{fig:algo}.
For each of these sets, we run our greedy clique formation algorithm {\gdycliques} on the respective induced subgame.
For $i\in [20]$, this yields a partial partition $\partition_i$ of $N^i$ together with a remainder set $R^i_1$.
We collect the union of the remainder sets as a set $R_1$.
We refer to this as \emph{Stage~1} of the algorithm.
Moreover, we refer to partitions produced by {\gdycliques} as \emph{clique partitions} and to the coalitions of a clique partition as a \emph{clique coalition}.

We continue with \emph{Stage~2} of the algorithm, where we cluster the clique partitions into larger coalitions that are individually rational and consist of exactly one coalition from each of the clique partitions.
This yields a partial partition $\partition$ as well as a further remainder set $R_2$.
The partition $\partition$ is indicated by the horizontal boxes in \Cref{fig:algo}.
We refer to partitions created by {\gdycluster} as \emph{merged partitions} and to the coalitions of a merged partition as a \emph{merged coalition}.
Finally, in \emph{Stage~3}, we add all agents in remainder sets to the coalitions of the partial partition $\partition$.
We will specify how exactly to select the respective coalitions in \Cref{spec:stagethree}.
This yields our final partition $\partition^*$.

\begin{algorithm}[tb]
  \caption{Computation of desirable partition in random ASHGs}\label{alg:main}
	\KwIn{Additively separable hedonic game $(N,u)$}
	\KwOut{Partition $\partition^*$}

Partition $N$ into $20$ sets $N^1,\dots, N^{\partnumber}$ with $|N^i| \in \left\{\left\lfloor \frac {|N|}{20}\right\rfloor, \left\lceil \frac {|N|}{20}\right\rceil\right\}$ for all $i\in [20]$  \\
\For{$i\in [20]$}
{
\tikzmark{ph1}$(\partition_i,R^i_1) \leftarrow \gdycliques(N^i,u|_{N^i})$%
}
$R_1 \leftarrow \bigcup_{i\in [20]}R^i_1$\\
\tikzmark{ph2}$(\partition,R_2)\leftarrow \gdycluster((N,u),(\partition_1,\dots,\partition_{\partnumber}))$\\ 
$\partition'\leftarrow \emptyset$,
$R\leftarrow R_1\cup R_2$\\
\For{$x\in R$}{
\tikzmark{ph3}
Select $C\in \partition$\\
$\partition' \leftarrow \partition' \cup \{C\cup\{x\}\}$,
$\partition\leftarrow\partition\setminus\{C\}$%
}
$\partition^*\leftarrow \partition\cup \partition'$\\
\KwRet {$\partition^*$}

\begin{tikzpicture}[remember picture, overlay]
\draw[thick, decoration={brace,raise=15pt},decorate]
  ([yshift=4.5ex,xshift = .8\textwidth]pic cs:ph3) -- node[right=20pt] (st3) {Stage $3$} ([yshift=-3.5ex,xshift = .8\textwidth]pic cs:ph3);
\draw[thick, decoration={brace,raise=15pt},decorate]
  ([yshift=4.5ex,xshift = .8\textwidth]pic cs:ph1) -- node[right=20pt] {Stage $1$} ([yshift=-3.5ex,xshift = .8\textwidth]pic cs:ph1);
\node at ([yshift = 8.5ex]st3) {Stage $2$};
\end{tikzpicture}

\end{algorithm}

\subsection{Stage 1: Greedy Clique Formation}

We now go through each of the stages of the algorithm separately and present the specific algorithms and theory.
We start with the details for the first stage of \Cref{alg:main}, where we execute the greedy clique formation procedure.
We essentially apply the greedy coloring algorithm by \citet{grimmett1975colouring}, which was designed for coloring {\ER} graphs.
Based, on the mutual utilities of agents, we assume that edges are present.
We then simply try to enlarge cliques until they reach a certain size, and remove them from the graph.
We iterate this until only a small fraction of vertices remains.
The pseudo code of this stage is presented in \Cref{alg:greedycol}.

\begin{algorithm}[tb]
\caption{\gdycliques}
	\label{alg:greedycol}
	\KwIn{Additively separable hedonic game $(N,u)$}
	\KwOut{Partition $\partition$ and remainder set $R$}
        $\partition \leftarrow \emptyset$\\
        $R \leftarrow V$\\
        \While{$R\neq \emptyset$}{
		Select $v \in R$\\
            $C \leftarrow \{v\}$\\
            $L \leftarrow R$\\
            \While{$L\neq \emptyset$ and $|C|<\ceil*{\frac{\log_{16}n}{2}}$}{
                Select $w \in L$\\
			\If{ $u_w(z)\ge \frac 12$ and $u_z(w) \ge \frac 12$ for every $z \in C$}{
				$C \leftarrow C\cup \{w\}$\\
			}
                $L \leftarrow L\setminus \{w\}$\\

		}
            \If{$|C|<\ceil*{\frac{\log_{16}n}{2}}$}{
                \KwRet $(\partition,R)$\\            
            }
            $R$ = $R\setminus C$\\
			$\partition $ = $\partition \cup \{\{C\}\}$
	}
 \KwRet $(\partition,R)$
\end{algorithm}

The critical property for the guarantee of \Cref{alg:greedycol} for random ASHGs can be extracted from the following insight about {\ER} graphs.
The proof of this result is analogous to the case $p=\frac{1}{2}$ as considered by \citet{grimmett1975colouring}.
We adapt the proof for dealing with general $p\in(0,1)$. 

\begin{restatable}{theorem}{maximalclique}\label{thm:greedycolouring}
Consider the {\ER} graph $G(n,p)$. Then, with probability at least $1-e^{-\Omega(\log_b^3 n)}$, every maximal clique of $G(n,p)$ has order greater than $\log_{b} n-3 \log_{b} \log_{b} n$ where $b=\frac{1}{1-p}$.
\end{restatable}

\begin{proof}
Consider a clique $C$ of size $k$ in $G(n,p)$.
Then, the probability that $C$ is maximal is $\left(1-(1-p)^k\right)^{n-k}$ because this is the probability that every vertex outside of $C$ has an edge to at least one vertex inside of $C$. 
As there are at most $\binom{n}{k} \le n^k$ cliques of size $k$, 
the expected number of maximal cliques of size at most $t=\log_{b} n -3\log_b\log_b(n)$ is at most 
\begin{align*}
&\sum_{k=1}^{\floor*{\log_b n -3 \log_b\log_b(n)}}\binom{n}{k}\left(1-(1-p)^k\right)^{n-k}\\&\leq\sum_{k=1}^{\floor*{\log_b n -3\log_b\log_b(n)}}n^k e^{-(n-k)(1-p)^{k}}\\
&\leq (\log_b n) n^{t} e^{-(1-\frac{t}{n})n(1-p)^t}\\
&\leq (\log_b n) n^{\log_b n} e^{-(1-\frac{t}{n})\log_b^3 n}\\
&=e^{-\Omega(\log_b^3 n)}
\end{align*}
since for $1\le k\le \log_b n -3 \log_b\log_b n$, the expression $n^k e^{-(n-k)(1-p)^{k}}$ is maximised if $k=\log_b n -3 \log_b\log_b n$.
Then, by Markov's inequality, the probability that we have at least one maximal clique of size at most $\log_b n -\log_b\log_b n$ is at most $e^{-\Omega(\log_b^3 n)}$.
\end{proof}

We can apply this insight to prove a guarantee for \Cref{alg:greedycol}.
Our analysis follows the proof by \citet{grimmett1975colouring} to show that their greedy coloring algorithm leads to a coloring with $\mathcal O\left(\frac n{\log_b n}\right)$ colors.

\begin{restatable}{theorem}{cliquepartition}\label{thm:cliquestage} 
Consider a random hedonic game $(N,u)\sim \randgame$ where $\distr = U(-1,1)$. 
Then, with probability $1 - e^{-\Omega(\log^3_b n)}$, \Cref{alg:greedycol} terminates with an output $(\partition, R)$ such that
\begin{itemize}
    \item $\partition$ is an $s$-balanced partition of $N\setminus R$ for $s = \ceil*{\frac{\log_{16}n}{2}}$ and
    \item $R$ is of size $|R| \le \frac n{\log^2_{16} n}$. 
\end{itemize}
\end{restatable}

\begin{proof}
    First, note that \Cref{alg:greedycol} can be interpreted as attempting to iteratively find a clique of size at least $\ceil*{\frac{\log_{16}n}{2}}$ in an {\ER} graph where the edge probability is $\frac 1{16}$.
    The latter is true because for the uniform distribution $\distr = U(-1,1)$ and for any pair of agents $w$ and $z$, it independently holds that $u_w(z)\ge \frac 12$ and $u_z(w)\ge \frac 12$ with probability $\frac 1{16}$.

    Now, assume that we reach the first while condition with a set $R$ with $n' = |R|$ with $n' \ge  \frac{n}{\log_b^2 n}$.
    By \Cref{thm:greedycolouring} applied to $G\left(n',\frac 1{16}\right)$, with probability $1-e^{-\Omega(\log_b^3 n')} = 1 - e^{-\Omega\left(\log_b^3 \left(\frac n{\log_b^2 n}\right)\right)} = 1-e^{-\Omega(\log_b^3 n)}$, there exists a clique of size at least 
    \begin{align*}\log_b n' &- 3 \log_b \log_b n'\geq \log_b (\frac{n}{\log_b^2 n}) - 3 \log_b \log_b (\frac{n}{\log_b^2 n})\\
    &=\log_b n-2\log_b\log_b  n -3 \log_b \log_b n+ 6 \log_b\log_b\log_b n \\
    &\geq \log_b n - 5 \log_b \log_b 
    n\\
    &\ge \ceil*{\frac{\log_{16}n}{2}}\text.\end{align*}

    Hence, by a union bound, the probability that we exit the while loop before $|R| < \frac{n}{\log_b^2 n}$ is bounded by $n e^{-\Omega(\log_b^3 n)}$, where we use the fact that \Cref{alg:greedycol} creates at most $n$ coalitions.
    In other words, \Cref{alg:greedycol} terminates with a remainder set of size smaller than $\frac{n}{\log_b^2 n}$ with probability at least $1 - n e^{-\Omega(\log_b^3 n)} = 1 - e^{-\Omega(\log^3_b n)}$, as desired.
\end{proof}

These properties motivate us for the following definition.
\begin{definition}\label{def:success:one}
    We say that \Cref{alg:main} \emph{succeeds in Stage~1} if for all $i\in [\partnumber]$ $\gdycliques(N^i,u|_{N^i})$ produces a pair of coalition and remainder set $(\partition_i,R^i_1)$ such that $\partition_i$ is an $s$-balanced partition of $N^i$ for $s = \ceil*{\frac{\log_{16}n}{2}}$ and $|R_1^i| \le \frac n{\log^2_{16} n}$.
\end{definition}

Note that \Cref{thm:cliquestage} immediately implies that \Cref{alg:main} succeeds in Stage~1 for random input games with probability $1 - 20 e^{-\Omega(\log^3_b n)} = 1 - e^{-\Omega(\log^3_b n)}$.
\subsection{Conditional Utilities after Stage 1}

In the analysis of our algorithm in \Cref{sec:properties}, we will have to compute the probability that the output satisfies certain properties.
However, these properties depend on the distribution of the utilities between the agents.
By the principle of deferred decisions \citep[see, e.g.,][Chapter~1.3]{MiUp05a}, the distribution of utilities that have not been revealed by the algorithm remains the original distribution, e.g., the uniform distribution $U(-1,1)$ if our game was sampled for $\distr = U(-1,1)$.

For the revealed utilities, the conditional distribution depends on the conditions that our algorithm checks for these utilities.
In this section, we want to establish the conditional distributions for the utilities revealed by the algorithm in Stage~1.

In \Cref{alg:greedycol}, we repeatedly check the condition whether a new agent $w$ can join a tentative coalition $C$, which is the case if $u_w(z)\ge \frac 12$ and $u_z(w) \ge \frac 12$ for every $z \in C$.
More specifically, we can implement this by checking agent pairs one by one. Once we reveal a utility of less than $\frac 12$, we do not reveal any further utility values and immediately exit the if-condition.
Revealing that an agent's utility for another agent is at least $\frac 12$ or less than $\frac 12$ leads to conditional utility distributions $U(0.5,1)$ and $U(-1,.5)$, respectively.
A visualization is given in \Cref{fig:revelation}.
We summarize this in the following specification for \Cref{alg:greedycol}.

\begin{specification}\label{fact:revelation}
    The inner if-condition of \Cref{alg:greedycol} can be implemented such that the following holds:
    For every coalition $C$ created by \Cref{alg:greedycol} and agent $a\in N\setminus C$, at most one revealed utility value of~$a$ for an agent in $C$ is less than $\frac12$, whereas all other revealed utilities are at least $\frac 12$.
    In particular, if $\distr = U(-1,1)$, then there exists a subset $C'\subseteq C$ with $|C'|\le 1$ such that the conditional utility after Stage~1 between $a$ and $z\in C$ is $U(-1,.5)$ if $z\in C'$,  $U(0.5,1)$ if $z\notin C'$ and $u_a(z)$ was revealed by \Cref{alg:greedycol}, and $U(-1,1)$ if $u_a(z)$ was not revealed by \Cref{alg:greedycol}. 
\end{specification}

\begin{figure}
    \centering
        \begin{tikzpicture}
            \node[draw, circle] (v1) at (0,4) {};
            \node[draw, circle] (v2) at (0,1.5) {};
            \node[draw, circle] (v3) at (0,-1) {};

            \node[draw, circle,label = 0:$\agtwo$] (w1) at (4,.6) {};
            \node[draw, circle,label = 0:$\agone$] (w2) at (4,3.4) {};

            \draw[bend right = 15,->] (v1) edge node[pos = .7, fill = white] {\footnotesize $< \frac 12$} (w2);
            \draw[bend right = 15,->] (w2) edge node[pos = .3, fill = white] {\footnotesize $\ge \frac 12$} (v1);
            \foreach \i in {3,2}{
            \draw[bend right = 12,->] (v\i) edge node[pos = .45, fill = white] {\footnotesize $\ge \frac 12$} (w1);
            \draw[bend right = 12,->] (w1) edge node[pos = .7, fill = white] {\footnotesize $\ge \frac 12$} (v\i);}
            \draw[bend right = 12,->] (v1) edge node[pos = .4, fill = white] {\footnotesize $\ge \frac 12$} (w1);
            \draw[bend right = 12,->] (w1) edge node[pos = .5, fill = white] {\footnotesize $\ge \frac 12$} (v1);

            \draw (0,1.5) ellipse (.9cm and 3cm);
            \node at (-1.1,1.5) {$C$};
        \end{tikzpicture}
    \caption{Visualization of utility revelation in Stage~1.
    Assume that $C$ is the tentative coalition.
    When checking whether $\agone$ can join $C$, we reveal a utility smaller than $\frac 12$ and exit the if-condition.
    For $\agtwo$, all utilities are large enough and therefore $\agtwo$ is added to $C$.}
    \label{fig:revelation}
\end{figure}
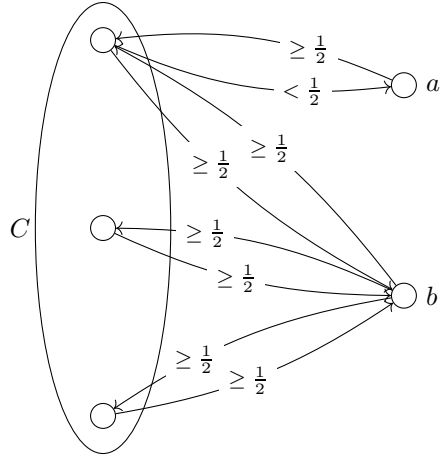

\subsection{Stage 2: Merging Coalitions via Greedy Clustering}

In the second stage of the algorithm, we aim at merging the $\partnumber$ partitions obtained in Stage~1 into larger coalitions.
Recall that we refer to the coalitions and partitions created by \Cref{alg:greedy} as merged coalitions and merged partitions, respectively.
The goal is to maintain individual rationality while enlarging coalitions to facilitate achieving {\envNoun}.
Hence, we want to merge coalitions if this operation does not decrease the utility of each involved agent by too much.

Assume that we are given a sequence of disjoint partial partitions $(\partition_1,
\dots, \partition_{\partnumber})$ in an random game $(N,u)$.
In \Cref{alg:main}, we will even have that each of these partitions is $s$-balanced for the same $s$.
Given a coalition $M\subseteq N$ and a coalition $C\in \partition_k$, we say that $C$ is \emph{compatible with} $M$ if, for all agents $\agone\in M$ and $\agtwo\in C$, it holds that $u_{\agone}(C) \ge -\frac 1{80}\ceil*{\frac{\log_{16}n}{2}}$ and $u_{\agtwo}(M) \ge -\frac {k-1}{80}\ceil*{\frac{\log_{16}n}{2}}$.
If this condition is not satisfied, then we say that $C$ is \emph{incompatible with} $M$.
If the coalition $M$ is clear from the context, we simply say that $C$ is \emph{incompatible}.
If \Cref{alg:greedy} executes the first if-condition with $C_1,\ldots C_{k-1}$ and $C$, then we call this a \emph{merge attempt}.
The merge attempt is said to be \emph{successful} if $C$ is compatible with $\bigcup_{i=1}^{k-1} C_i$, and unsuccessful, otherwise.

\begin{algorithm}[tb]
  \caption{\gdycluster}
  \label{alg:greedy}
	\KwIn{Additively separable hedonic game $(N,u)$ together with sequence $(\partition_1,
\dots, \partition_{\partnumber})$ of disjoint partial partitions of~$N$}
	\KwOut{Partial partition $\partition$ and remainder set $R$
 }

$\partition \leftarrow \emptyset$\\
\While {$\partition_1\neq \emptyset$}{
$C_k \leftarrow\emptyset$ for all $k\in [\partnumber]$\\
\For {$k = 1,\dots, \partnumber$}{
\For {$C\in \partition_k$}
{\If {$C$ is compatible with $\bigcup_{i=1}^{k-1} C_i$}
{
$C_k\leftarrow C$\\
\textbf{break}
}
} 

} 

\If {$C_k = \emptyset$ for some $k\in [\partnumber]$}{
    $R\leftarrow \bigcup_{k = 1}^{\partnumber} \bigcup_{C\in \partition_k}C$\\
    \KwRet $(\partition,R)$
}
{$\partition_k \leftarrow \partition_k \setminus \{C_k\}$ for all $k\in [\partnumber]$\\
$\partition \leftarrow \partition \cup \left\{\bigcup_{k=1}^{\partnumber}C_k\right\}$
}
}
$R\leftarrow \bigcup_{k = 1}^{\partnumber} \bigcup_{C\in \partition_k}C$\\
\KwRet $(\partition,R)$
\end{algorithm}

In \Cref{alg:greedy}, we seek to merge one coalition from each of the partitions $\partition_i$, $i\in [20]$, successfully.
Given coalitions $C_i\in \partition_i$ for each $i \in[20]$, we say that $(C_1,\dots, C_{\partnumber})$ \emph{merge successfully} if, for each $2\le \ell \le \partnumber$, it holds that $C_{\ell}$ is compatible with $\bigcup_{j = 1}^{\ell-1}C_j$.
\Cref{alg:greedy} attempts to find sequences of coalitions that merge successfully. If such a sequence is found, the algorithm adds the merged coalition $\bigcup_{j = 1}^{\partnumber}C_j$ to its output partition.
The algorithm will try to create a merged partition starting from each $C\in \partition_1$.
Once it does not succeed to merge a partial coalition $(C_1,\ldots,C_i)$ with any remaining $C_{i+1}\in \partition_{i+1}$, \Cref{alg:greedy} terminates and outputs the partial partition created until then as well as the set of remaining agents.
At this point, Stage~2 in \Cref{alg:main} terminates.

\paragraph{Notation}
Before we succeed with formal results about Stage~2, we define some useful notation to easily refer
to the objects that an agent is part of during the execution of \Cref{alg:main}.
Let $(N,u)$ be an additively separable hedonic game and consider an agent $a\in N$.
Assume that we run \Cref{alg:main} for $(N,u)$.
First, we denote by $\ind(\agone)$ the index of the unique subset of $N$ produced in the preprocessing such that $\agone\in N^{\ind(\agone)}$.
Moreover, if $\agone\notin R_1$ at the end of Stage~1 of \Cref{alg:main}, we denote by $\coa(\agone)$ the clique coalition that $\agone$ is part of, i.e., $\coa(\agone) := \partition_{\ind(\agone)}(\agone)$.
If $\agone\notin R_1\cup R_2$ at the end of Stage~2 of \Cref{alg:main}, we denote by $\mcoa(\agone)$ the merged coalition that $x$ is part of.
Finally, for $i\in \mathbb N$ with $i\ge 2$ consider a sequence of sets $(S_1, \ldots, S_i)$. 
We define $Pairs(S_1,\ldots,S_i):=\{(x,S_j)\mid x\in S_i, j\in [i-1]\}\cup \bigcup_{j\in [i-1]}\{(x,S_i)\mid x\in S_j\}$.
We will apply this notation for sequences of coalitions $(C_1,\ldots,C_k)$ where $C_j \in \partition_j$ for all  $j\in [k]$ and $2\le k\le \partnumber$.
It is straightforward to see that we can bound $Pairs(C_1,\ldots,C_i)$ for equally sized coalitions.

\begin{restatable}{lemma}{pairscount}\label{prop:pairscount} Let $2 \le i \le \partnumber$ and, for each $j\in [i]$, let $C_j\in \partition_j$ with $|C_j| = \ceil*{\frac{\log_{16} n}{2}}$. Then, it holds that
$$|Pairs(C_1,\ldots,C_i)| 
\leq 38\ceil*{\frac{\log_{16} n}{2}}\text.$$
\end{restatable}

\begin{proof}
    Let $2 \le i \le \partnumber$ and, for each $j\in [i]$, let $C_j\in \partition_j$ with $|C_j| = \ceil*{\frac{\log_{16} n}{2}}$.
    Then,
    $$|Pairs(C_1,\ldots,C_i)| = (i-1)|C_i| + \sum_{j=1}^{i-1}|C_j| = 2(i-1)|C_i| =  2(i-1)\ceil*{\frac{\log_{16} n}{2}}\leq 38\ceil*{\frac{\log_{16} n}{2}}\text.$$
\end{proof}

We now move on to proving formal high probability claims about the partition produced in Stage~2 of \Cref{alg:main}. 
A crucial observation is that, while Stage~1 only depends on edges between agents of a fixed set $N^i$ for some $i\in [\partnumber]$, Stage~2 only depends on edges between two different such agent sets.
Hence, by the principle of deferred decisions, the distribution of utilities between agents from different sets $N^i$ is independent from Stage~1 when we enter Stage~2.
Thus, these utilities are still distributed the same as in the random game with which we started the algorithm. 
Our first goal is to give a bound on the probability that a merge attempt fails.
We again apply Hoeffding's inequality (cf. \Cref{thm:Hoeffding}).

\begin{restatable}{lemma}{LemMerge}
    \label{lem:merge} 
Let $\distr = U(-1,1)$ and consider a random hedonic game $(N,u)\sim \randgame$.
Then, there exists a constant $0<r<1$, such that, for every $2 \le i \le \partnumber$ and $C_j\in \partition_j$ with $|C_j| = \ceil*{\frac{\log_{16} n}{2}}$ for $j\in [i]$, it holds that
 $$\Prob\left( C_i \text{ incompatible with } \bigcup_{j=1}^{i-1}C_j\right)\leq \frac{1}{n^r}$$ for large $n$. 
\end{restatable}

\begin{proof}
Let $\agone\in N$ and $C\subseteq N\setminus \{\agone\}$ with $|C| = \ceil*{\frac{\log_{16} n}{2}}$.
We will first 
show that there exists a constant $0< r' < 1$ such that $\Prob(u_a(C)\leq -\frac{t}{80})\leq \frac{1}{n^{r'}}$.

Define $t := \ceil*{\frac{\log_{16} n}{2}}$ and consider random variables $Y_1,\ldots, Y_{t}$ where $Y_i \sim \mathcal D$. 
Then, $u_{\agone}(C)$ is distributed according to $\sum_{i=1}^t Y_i$. 
Note that $-1\leq Y_i\leq 1$ and $\Ex[u_{\agone}(C)]=0$.
By \citeauthor{Hoef63a}'s inequality (\Cref{thm:Hoeffding}), we obtain 
\begin{align*}
&\Prob\left(u_{\agone}(C)\leq -\frac{t}{80}\right)=\Prob\left(u_{\agone}(C)\geq \frac{t}{80}\right)\\
&=\Prob\left(u_{\agone}(C)-\Ex[u_{\agone}(C)]\geq \frac{t}{80}\right)\leq e^{-\frac{2\cdot \left(\frac{t}{80}\right)^2}{4\cdot t}}\\
&=e^{-\frac{t}{12800}}=e^{-\frac{\ceil*{\frac{\log_{16} n}{2}}}{12800}}\leq e^{-\frac{\log_{16} n}{25600}}=\frac{1}{n}^{\frac{1}{\ln 16 
\cdot 25600}}\text.\end{align*}
Hence, our claim is true for $r'=\frac{1}{\ln 16 
\cdot 25600}$, which satisfies $0<r'<1$. 
Now, observe that, by definition of compatibility, incompatibility can only happen if there exists $(\agone,C)\in Pairs(C_1,\ldots,C_i)$ with $u_{\agone}(C) \le -\frac{t}{80}$.
We apply \Cref{prop:pairscount} and a union bound to conclude
\begin{align*}
    \Prob\left( C_i \text{ incompatible with } \bigcup_{j=1}^{i-1}C_i\right)\leq \sum_{(x,C)\in Pairs(C_1,\ldots,C_i)} \Prob\left(u_{\agone}(C)\leq-\frac{t}{80}\right)
     \leq 38\ceil*{\frac{\log_{16} n}{2}}\frac 1{n^{r'}}\text.
\end{align*}

Clearly, there exists $0<r<1$ such that, for $n$ large enough, it holds that 
$38\ceil*{\frac{\log_{16} n}{2}}\frac 1{n^{r'}} \le \frac{1}{n^r}$. 
This completes the proof.
\end{proof}

Similar to \Cref{def:success:one} for Stage~1, we define success of Stage~2.
\begin{definition}\label{def:success:two}
    We say that \Cref{alg:main} \emph{succeeds in Stage~2} if Stage~2 terminates with $(\pi, R_2)$ where $\pi$ is $s$-balanced for $s = \partnumber \ceil*{\frac{\log_{16}  n}{2}}$ and $|R_1\cup R_2|\le 20 \frac n{\log^2_{16} n} + \frac{80}{r}\cdot \ceil*{\frac{\log_{16}n}{2}}$ where $r$ is the parameter obtained in \Cref{lem:merge}.
\end{definition}

We are ready to prove a success guarantee for Stage~2.

\begin{restatable}{theorem}{mergingstage}\label{thm:mergingstage}
Assume that we run \Cref{alg:main} for a random input game $(N,u)\sim H(n,\mathcal{D})$.
Assume that \Cref{alg:main} succeeds in Stage~1. 
Then, with probability $1-n^{-3}$ \Cref{alg:main} succeeds in Stage 2.

\end{restatable}

\begin{proof}
Consider the event $\mathcal A$ defined as ``Stage~2 terminates with the following property: $R_2$ contains agents from at most $\frac 4 r$ clique coalitions from each of the clique partitions $\partition_1,\dots, \partition_{\partnumber}$ where $r$ is the parameter obtained in \Cref{lem:merge}.''
Our goal is to prove that $\mathcal A$ occurs with high probability.
Note that, if Stage~1 succeeds, each partial partition created in Stage~1 contains the same number of coalitions.
Hence, the number of remaining coalitions is the same for each partial partition. 
Thus, event $\mathcal A$ occurs whenever $R_2$ contains agents from at most $\frac 4r$ clique coalitions from \emph{some} clique partition.

Now let $2 \le i\le 20$ and consider any time during the execution of \Cref{alg:greedycol} where we want to find a coalition $C\in \partition_i$ compatible with $\bigcup_{j = 1}^{i-1}C_j$ while $|\partition_i|\ge \frac 4 r$ clique coalitions.
Since each of the merge attempts depends on a disjoint set of utility values we can apply \Cref{lem:merge} to deduce
$$\Prob\left(\bigcup_{j = 1}^{i-1}C_j \text{ incompatible with all }C\in \partition_i\right) \le \left(\frac{1}{n^r}\right)^{\frac{4}{r}} =  \frac{1}{n^4}\text.$$

Now, Stage~2 can only terminate if some tentative coalition is incompatible with all proposed coalitions.
In total, $\partition$ contains at most $\frac n {\partnumber}$ coalitions\footnote{For simplicity, we estimate with the number of agents here, whereas we know that coalitions are of size $\ceil*{\frac{\log_{16}n}{2}}$ and therefore there are only $\mathcal O\left(\frac n{\log_{16}n}\right)$ coalitions in $\partition_1$.} and each of them is only attempted to be increased by a coalition from the next partition for $\partnumber$ times (even 1 less times).
For the complement $\mathcal A^c$ of $\mathcal A$, we obtain
$$\Prob(\mathcal A^c) = \Prob\left(\text{some }\bigcup_{j = 1}^{i-1}C_j \text{ incompatible with all }C\in \partition_i \text{ while } |\partition_i|\ge \frac 4r\right) \le n\cdot \frac 1{n^4} = \frac 1 {n^3}\text.$$

For the inequality, we apply a union bound for the at most $n$ events. 
Hence, $\Prob(\mathcal A)\ge 1 - \frac 1{n^3}$.

For the case that $\mathcal A$ occurs, let $i\in [20]$ such that at most $\frac 4r$ coalitions in $\partition_i$ remain at the end of Stage~2.
Since Stage~1 succeeds, we also know that $|R_1^i|\le \frac n{\log^2_{16} n}$.
Together, $$|R\cap N^i| \le |R_1^i| + \frac 4r \cdot \ceil*{\frac{\log_{16}n}{2}} \le \frac n{\log^2_{16} n} + \frac 4r \cdot \ceil*{\frac{\log_{16}n}{2}}\text.$$
Since every merged coalition contains exactly the same number of agents from each of the agent sets $N^i$, we can conclude that 
$$|R| = \sum_{i\in [\partnumber]}|R\cap N^i| \le 20 \frac n{\log^2_{16} n} + \frac{80}{r}\cdot \ceil*{\frac{\log_{16}n}{2}}\text.$$
This concludes the proof because our desired properties are satisfied whenever event $\mathcal A$ occurs.
\end{proof}

\subsection{Stage 3: Completing the Partition}

\Cref{alg:main} attempts to complete the partition $\partition$ by assigning agents from sets $R_1\cup R_2$ to distinct merged coalitions in $\partition$, thereby producing the final partition $\partition^*$.
This leaves still some choice to select an appropriate coalition.
For the individual rationality of the produced partitions, we want to guarantee that each agents receives a positive utility.
Moreover, to simplify our analysis of {\exvNoun}, it is beneficial to add agents only to coalitions towards which no utility was revealed in Stage~2.

\begin{restatable}{theorem}{finalstage}\label{thm:finalstage}
Assume that Stage~2 of \Cref{alg:main} succeeds on input $(N,u)\sim H(n,\mathcal{D})$.
Then, with probability $1 - \mathcal O\left(\frac 1{n^3}\right)$ each agent in $\agone\in R$ can be added to a coalition $C$ in Stage~3 such that 
\begin{itemize}
    \item no utilities between $\agone$ and $C$ were revealed in Stage~2 and
    \item $u_{\agone}(C\cup \{\agone\}) >0$.
\end{itemize}
\end{restatable}

\begin{proof}
If \Cref{alg:main} was successful in Stage~2, then, by definition, it terminates Stage~2 with a partial partition $\partition$ containing $|\partition| = \Theta\left(\frac{n}{\log n}\right)$ coalitions and a remainder set $R$ of size $|R|=\mathcal O\left(\frac{n}{\log^2 n}\right)$.
Hence, for large $n$, it holds that $|\partition|$ is much larger than $3|R|$.
Therefore, it suffices to show that every agent in $R$ receives a positive utility when joining at least $\frac 13$ of the merged coalitions for which no utility was revealed in Stage~2.
Then, all agents in $R$ can be added to a coalition with the desired properties.

Consider an agent $\agone\in R$.
If $a\in R_1$, then no utilities for $a$ were revealed in Stage~2.
Hence, assume that $a\in R_2$.
Then, $\coa(a)$ failed to be compatible in any performed merge attempt in Stage~2.
By \Cref{lem:merge} the probability of one such failed merge attempt is bounded by $\frac{1}{n^r}$ for some $r$.
Hence, the probability that $\coa(a)$ failed to be merged for more than $\frac{4}{r}$ times is bounded by $\frac{1}{n^r}^{\frac{4}{r}}\leq \frac{1}{n^4}$.
Hence, the probability that information regarding $a$'s utility for more than $\frac{4}{r}$ merged coalitions in $\partition$ was revealed during the execution of the algorithm is at most $\frac{1}{n^4}$.

We consider again a general agent $\agone\in R$.
    Let $\partition^a\subseteq \partition$ be the subset of merged coalitions such that no utility information concerning these coalitions and agent $a$ was revealed in Stage~2.
    By the previous argument, it holds that 
    \begin{equation}\label{eq:mergesuccesses}
        \Prob\left(|\partition\setminus \partition^a|\le \frac{4}{r}\right) \ge 1-n^{-4}\text.
    \end{equation}
    
    Now consider any coalition $C\in \partition^a$. 
    Let $\chi_{C,a}$ be the indicator random variable that is $1$ if $u_a(C) > 0$ and $0$, otherwise. 
    Agent $a$'s conditional utility distribution for agents in $C$ at the end of Stage~2 can only be different from $U(-1,1)$ for revealed utilities in Stage~1 towards agents in $C\cap N^{\ind(a)}$.
    By \Cref{fact:revelation}, for all except one of the revealed utilities, the distribution is $U\left(\frac 12,1\right)$ whereas the distribution for at most one agent is $U\left(-1,\frac 12\right)$. 
    Let $t$ be the number of agents in $C$ to which no utility information was revealed.
    Then, $t\ge \left|C\setminus N^{\ind(a)}\right| = 19 \ceil*{\frac{\log_{16} n}{2}}$.
    Hence, if $Y_1,\ldots,Y_t$ are independent random variables distributed according to $U(-1,1)$, then $u_a(C)\ge \sum_{j =1}^t Y_j - 1$, where the subtraction of $1$ accounts for the one possible agent with distribution $U\left(-1,\frac 12\right)$ and we omit all other revealed utilities, which have to be positive.
    Hence, for large enough $n$, it holds that $\Prob(\chi_{C,a} = 1) = \Prob(u_a(C) > 0)\ge \frac 49$.
    
    We apply a Chernoff bound to that at least a third of the coalitions in $\partition^a$ yield a positive utility to $a$.
    Let $\delta=\frac{4}{9}$.
    Then,

    \begin{align*}
    &\Prob\left(\frac{1}{|\partition^a|}\sum_{C\in \partition^a}\chi_{C,a}\leq \frac{1}{3}\right)=
    \Prob\left(\frac{1}{|\partition^a|}\sum_{C\in \partition^a}\chi_{C,a}\leq (1-\delta)\frac{3}{5}\right)\\
    &=
    \Prob\left(\sum_{C\in \partition^a}\chi_{C,a}\leq (1-\delta)\frac{3}{5}\left|\partition^a\right|\right)\leq e^{-\frac{\delta^2\frac{3}{5}\left|\partition^a\right|}{2}}\text. 
    \end{align*}

    Now assume that Stage~3 of \Cref{alg:main} is implemented such that each remainder agent is added to some coalition in $\partition^a$ that maximizes their utility among the remaining coalitions.
    Assume that this yields the outcome $\partition^*$.

    \begin{align}
        &\Prob\left(u_a(\partition^*)>0\right)\notag\\
        &\ge 
        \Prob\left(u_a(\partition^*)>0\ \Big|\  |\partition\setminus \partition^a|\le \frac{4}{r}\right)\cdot \Prob\left(|\partition\setminus \partition^a|\le \frac{4}{r}\right)\label{eq:mergesuccessapp}\\
        &\ge \Prob\left(\frac{1}{|\partition^a|}\sum_{C\in \partition^a}\chi_{C,a} > \frac{1}{3} \ \Big|\ |\partition\setminus \partition^a|\le \frac 4r\right)\cdot \Prob\left(|\partition\setminus \partition^a|\le \frac{4}{r}\right)\label{eq:onethird}\\
        &\ge \left(1-e^{-\Omega\left(\frac{n}{\log_{16} n}\right)}\right)\cdot \left(1-\frac 1{n^4}\right) = 1- \mathcal O\left(\frac 1{n^4}\right)\label{eq:pia}\text.
    \end{align}
    
    For Inequality (\ref{eq:mergesuccessapp})     we use the law of total probability and apply Inequality \ref{eq:mergesuccesses}.
    For Inequality (\ref{eq:onethird}), we use that the utility is positive if at least one third of the coalitions in $\partition^a$ yield a positive utility.
    Above, we have argued that this is the case for $\partition$, but it also holds for $\partition^a$ if its size only differs from the size of $\partition$ by a constant.
    For Inequality (\ref{eq:pia}), we use that if $|\partition\setminus \partition^a|\le \frac{4}{r}$, then $|\partition^a|\geq|\partition|-\frac{4}{r}=\Omega\left(\frac{n}{\log_{16} n}\right)$
    because $r$ is a fixed constant.

By a union bound, since $|R|\le n$, every agent received positive utility in $\partition^*$ with probability at least $1- \mathcal O\left(\frac{1}{n^3}\right)$.
This completes the proof because $\partition^*$ was constructed so that agents are only added to coalitions for which no information was revealed in Stage~2.
\end{proof}

Based on this proof, we make the following specification for implementing Stage~3.

\begin{specification}\label{spec:stagethree}
    In Stage~3, when selecting a coalition for a remainder agent $\agone$, \Cref{alg:main} selects a coalition in which $\agone$ receives positive utility and for which no utilities were revealed in Stage~2.
\end{specification}

We define success of Stage~3.

\begin{definition}
    We say that \Cref{alg:main} \emph{succeeds in Stage~3} if it computes a valid coalition structure when Stage~3 is executed according to \Cref{spec:stagethree}.
\end{definition}

\Cref{thm:finalstage} implies that success in Stage~3 is once again a high probability event.

\section{Properties of Our Algorithm}\label{sec:properties}
In \Cref{sec:2stageclust}, we have described \Cref{alg:main} and have shown structural results for its output that apply with high probability.
We leverage these insights to show that the algorithm produces partitions satisfying extensive stability properties.
Specifically, we prove that, with high probability, the output produced by \Cref{alg:main} is
\begin{enumerate}
    \item individually rational (\Cref{sec:algIR}),
    \item {\env} (\Cref{sec:ENV}), and
    \item {\exv} (\Cref{sec:EXV}).
\end{enumerate}

Given a game $(N,u)$, let $\alg(N,u)$ denote the partition produced by \Cref{alg:main}.
Let $\mathfrak P$ be any property among the properties individual rationality, {\envNoun}, and {\exvNoun}.
Our goal is to show that for each of these properties
\begin{equation*}%
    \Prob\left(\alg(N,u) \text{ satisfies } \mathfrak P\right) \to 1 \text{ for } n\to \infty\text.
\end{equation*}

To show this, we condition on the high probability event that Stage~2 succeeds (cf. \Cref{def:success:two}).
Given a hedonic game $(N,u)$, let $\Stwo$ denote the event that Stage~2 of \Cref{alg:main} succeeds on input $(N,u)$.
We have that 
\begin{align*}
    &\Prob\left(\alg(N,u) \text{ satisfies } \mathfrak P\right)\\
    &\ge \Prob\left(\alg(N,u) \text{ satisfies } \mathfrak P \mid \Stwo \right) \cdot\Prob\left(\Stwo\right)\text.
\end{align*}

By \Cref{thm:finalstage}, we already know that $\Prob\left(\Stwo\right)$ tends to $1$ in the large agent limit.
Hence, it suffices to show that
\begin{equation}\label{eq:highcondprob}
    \Prob\left(\alg(N,u) \text{ satisfies } \mathfrak P \mid \Prob\left(\Stwo\right)\right)\to 1 \text{ for } n\to \infty\text.
\end{equation}

\subsection{Individual Rationality}\label{sec:algIR}

Most of the work to establish individual rationality was already done in \Cref{thm:finalstage}.
It implies that agents in the remainder set are individually rational with high probability.
We only need to prove individual rationality of agents that are part of a clique coalition after Stage~1 of the algorithm.
However, then individual rationality follows directly from the positive utility by agents in their own clique coalition.

\begin{restatable}{lemma}{notr}\label{lem:notr} 
Consider a random hedonic game $(N,u)\sim H(\mathcal{D},n)$ where $\distr = U(-1,1)$. 
Then, for $n$ large enough, it holds that $$\Prob(u_a(\alg(N,u))>0\mid a\notin R,\Stwo) = 1.$$
There, $R$ is the random remainder set resulting from running \Cref{alg:main} for $(N,u)$.
\end{restatable}

\begin{proof}
For $\distr = U(-1,1)$, consider a hedonic game $(N,u)\sim H(\mathcal{D},n)$. 
Define $t := \ceil*{\frac{\log_{16}n}2}$.
Let $\partition^* = \alg(N,u)$ and let $\partition$ be the partial partition consisting of merged coalitions obtained after Stage~2.
If \Cref{alg:main} succeeds in Stage~2 on input $(N,u)$, then $\partition$ is $s$-balanced for $s = 20t$.

Let $a\in N\setminus R$.
Then, since utilities are bounded from below by $-1$ and since every coalition in $\partition^*$ contains at most one more member as compared to $\partition$, it holds that $u_a(\partition^*)\geq u_a(\partition)-1$.  

Moreover, recall that the merged coalition $\mcoa(a)$ of agent $a$ in $\partition$ is the disjoint union of $\partnumber$ clique coalitions $C_{1},\ldots, C_{\partnumber}$.
By design of \Cref{alg:main}, every member of $\coa(a) = C_{\ind(a)}$ yields a utility distributed as $U(0.5.1)$ to $a$.
Hence, in particular $u_a(C_{\ind(a)})\geq \frac{t-1}{2}$.
Moreover, because $(C_1,\dots, C_{\partnumber})$ merge successfully, it holds that%
$\sum_{j = 1, j\neq \ind(a)}^{20}u_a(C_j)\geq -\frac{19t}{80}$.
Together, 
\begin{align*}u_a(\partition^*)\ge u_a(\mcoa(a))-1=\sum_{j=1}^{\partnumber} u_a(C_j) - 1\geq \frac{t-1}{2}-\frac{19t}{80} - 1
\geq \frac{t}{4}-\frac{3}{2}.\end{align*}
Clearly, for sufficiently large $n$, this utility is strictly positive.
This completes the proof.
\end{proof}

Combining \Cref{lem:notr} with \Cref{thm:finalstage}, we obtain that \Cref{alg:main} produces an individually rational partition with high probability.
\begin{restatable}{theorem}{propertyIR}\label{thm:ir}
Consider random hedonic games for $\distr = U(-1,1)$. Then,
\Cref{alg:main} produces an individually rational partition with probability tending to~$1$ in the large agent limit.
\end{restatable}

\begin{proof}
    We will prove \Cref{eq:highcondprob} for the case where $\mathfrak P$ represents individual rationality.
    Consider an agent $a\in N$.
    Let $\mathcal B_a$ denote the event that $u_a(\alg(N,u)) \le 0$.
    By \Cref{thm:finalstage}, it holds that 
    $$\Prob\left(\bigcup_{a\in R} \mathcal B_a \ \Big|\ \Stwo\right) = \mathcal O\left(\frac 1{n^3}\right)\text.$$

    Moreover, by \Cref{lem:notr}, it holds that 
    $$\Prob\left(\bigcup_{a\notin R} \mathcal B_a \ \Big|\ \Stwo\right) = 0\text.$$

    Together, 
    \begin{align*}
        &\Prob\left(\alg(N,u) \text{ individually rational} \mid \Stwo\right)\\
        &= 1 - \Prob\left(\bigcup_{a\in N} \mathcal B_a \ \Big|\ \Stwo\right)\\
        &= 1 - \Prob\left(\bigcup_{a\notin R} \mathcal B_a \ \Big|\ \Stwo\right) - \Prob\left(\bigcup_{a\in R} \mathcal B_a \ \Big|\ \Stwo\right) \\ 
        &= 1 - \mathcal O\left(\frac 1{n^3}\right)\text.
    \end{align*}

    Clearly, this expression tends to $1$ for $n$ tending to infinity.
\end{proof}

\subsection{\ENVNoun}\label{sec:ENV}

In this section, we will show that \Cref{alg:main} produces partitions satisfying {\envNoun} with high probability.
Recall that, given a coalition $C\subseteq N$ and an agent $a\in N\setminus C$, $\FavorOut(C,a)$ denotes the so-called favor-out set of agents, i.e., the set of agents whose utility is decreased if $a$ joins them.
We want to show that the probability that the event $\FavorOut(C,a) = \emptyset$ occurs for any coalition $C\in \alg(N,u)$ and $a\in N\setminus C$ tends to $0$.
Since there are at most $n^2$ pairs of coalitions and agents, it suffices to show that 
\begin{equation*}
    \Prob\left(\OutEvent\ \big|\ \Stwo\right) = o\left(n^{-2}\right)\text.
\end{equation*}
The claim then follows by a union bound.

Intuitively, $\OutEvent$ should have a small probability because, for the unconditional utility distribution according to $U(-1,1)$, each agent in coalition $C$ prefers agent $a$ not to join with probability $\frac 12$.
Our goal is to show that the probability is still sufficiently small for the conditional utility distribution after running \Cref{alg:main}.
Recall that we may reveal utilities concerning $a$ and $C$ in three ways: 
\begin{enumerate}
    \item In Stage~1, we may reveal utilities between $a$ and agents in $C\cap N^{\ind(a)}$.
    \item In Stage~2, we may reveal utilities during merge attempts. However, since $a\notin C$, we only care about utilities revealed during \emph{unsuccessful} merge attempts.
    \item In Stage~3, we may reveal utilities of $a$ for agents in $C$ if $a\in R$, i.e., $a$ is in the remainder set. 
\end{enumerate}

To deal with the first case, we show that the probability of $\OutEvent$ is already sufficiently low if we just consider the possible denial by agents in $C\setminus N^{\ind(a)}$, for which Stage~1 does not affect the conditional utility.

Similarly, for the third case, if $C$ contains an agent from $R$, we also disregard their potential denial. 
The crucial insight is how to deal with the second case.
Intuitively, an unsuccessful merge event should only decrease the probability that all agents in a coalition have a nonnegative utility for some agent.
We will turn this intuition into a rigorous formal statement by applying an FKG inequality (\citealp{FKG71a}, see also \citealp{Grim99a}, Theorem 2.4).

\begin{lemma}[\citealp{FKG71a}]\label{lem:harris} Let $f$ and $g$ be functions of 
random variables with $\mathbb E[f^2] < \infty$ and $\mathbb E[g^2] < \infty$.
If $f$ is monotone increasing and $g$ is monotone decreasing
then $$\mathbb{E}[fg]\leq \mathbb{E}[f]\mathbb{E}[g]$$
\end{lemma}

We want to apply the inequality for the case where $f$ is the indicator function of certain agents not vetoing an agent to join them and $g$ is the indicator function of certain unsuccessful merges.

\begin{restatable}{lemma}{LemCorrelation}\label{lem:correlation} Consider a random hedonic game $(N,u)\sim \randgame$ where $\distr = U(-1,1)$.
Let $C\in \alg(N,u)$ and $a\in N\setminus C$. 
Then, $$\Prob\left(\OutEvent\ \big|\ \Stwo\right)\le n^{-\frac {19}8}\text.$$ 
\end{restatable}

\begin{proof}
Consider a random hedonic game $(N,u)\sim \randgame$ where $\distr = U(-1,1)$.
Let $C\in \alg(N,u)$ and $a\in N\setminus C$.

The key insight is to prove the correlation of failed merge events with {\envNoun} by applying \Cref{lem:harris}.
We start by precisely defining the functions, for which we want to apply it.

First, for a pair of agents $a,b\in N$ with $a\neq b$, let $\PosEvent$ be the event that $u_b(a)\ge 0$, i.e., $a$ is tolerated by $b$ to join their coalition.
Note that for $b, c\in N\setminus \{a\}$ with $b\neq c$, the events $\PosEvent$ and $\PosEvent[a,c]$ depend on two different, independently sampled utilities, and these events are therefore independent.
Let $f_{a,b}$ be the indicator random variable of the event $\PosEvent$.
For a subset $B\subseteq N\setminus \{a\}$, define $f_B := \prod_{b\in B} f_{a,b}$.

Moreover, for $2\le k\le \partnumber$, let $F = (C_1,\dots, C_k)$ denote a sequence of coalitions where, for $i\in [k]$, we assume $C_i\in \partition^i$.
Let $\FailEvent$ be the event that the merge attempt of $(C_1,\dots, C_{k-1})$ with $C_k$ is not successful and let $g_F$ be the indicator random variable of $\FailEvent$. 
For a set of merge attempts $\mathfrak F = \{F = (C_1,\dots, C_k)\colon 2\le k\le \partnumber, C_i\in \partition^i \text{ for } i\in [k]\}$, define the function $g_{\mathfrak F} := \prod_{F\in \mathfrak F}g_F$.

Clearly, for all $a, b\in N$ with $a\neq b$, $f_{a,b}$ is an increasing function on the space of utility profiles because increasing utilities can only turn an agent objecting another agent to join to not objecting them anymore and not vice versa. Hence, as a product of nonnegative, increasing functions, $f_M$ is increasing.
Similarly, for a sequence of coalitions $F$, $g_F$ is a decreasing function on the space of utility profiles because decreasing utilities can only turn a successful merge into an unsuccessful merge and not vice versa.
Hence, as a product of nonnegative, decreasing functions, $g_{\mathfrak F}$ is a decreasing function.
Additionally, since both of these functions are indicator functions, it holds that $\mathbb E[f^2] < \infty$ and $\mathbb E[g^2] < \infty$.
Hence, by \Cref{lem:harris}, it holds that
\begin{equation}\label{eq:correlation}
    \mathbb E[f_B\cdot g_{\mathfrak F}]\le \mathbb E[f_B]\mathbb E[g_{\mathfrak F}]\text.
\end{equation}

Note that it holds that $\Prob\left(\bigcap_{b\in B} \PosEvent\right) = \mathbb E[f_B]$, $\Prob\left(\bigcap_{F\in \mathfrak F} \FailEvent\right) = \mathbb E[g_{\mathfrak F}]$, as well as $\Prob\left(\bigcap_{b\in B} \PosEvent\cap \bigcap_{F\in \mathfrak F} \FailEvent\right) = \mathbb E[f_B\cdot g_{\mathfrak F}]$.
Hence, \Cref{eq:correlation} directly implies that 
\begin{equation}\label{eq:HarrisConditional}
    \Prob\left(\bigcap_{b\in B} \PosEvent\ \bigg|\ \bigcap_{F\in \mathfrak F} \FailEvent\right) 
    = \frac{\Prob\left(\bigcap_{b\in B} \PosEvent\cap \bigcap_{F\in \mathfrak F} \FailEvent\right)}{\Prob\left(\bigcap_{F\in \mathfrak F} \FailEvent\right)}
    \le \Prob\left(\bigcap_{b\in B} \PosEvent\right)\text.
\end{equation}

We can now apply \Cref{eq:HarrisConditional} to prove the assertion of the lemma.
Consider a random hedonic game $(N,u)\sim \randgame$ where $\distr = U(-1,1)$.
Let $C\in \alg(N,u)$ and $a\in N\setminus C$.
If Stage~2 succeeds, then $C$ contains $19\ceil*{\frac{\log_{16} n}{2}}$ agents in $N\setminus N^{\ind(a)}$.
Consider $B = C\setminus N^{\ind(a)}$ and
$\mathfrak F$ be the set of unsuccessful merge events in the execution of \Cref{alg:main} for $(N,u)$.
Then,

\begin{align*}
    &\Prob\left(\OutEvent\ \big|\ \Stwo\right) = \Prob\left(\bigcap_{b\in C} \PosEvent\ \Big|\ \Stwo\right)\\
    &\le \Prob\left(\bigcap_{b\in B} \PosEvent\ \Big|\ \Stwo\right)= \Prob\left(\bigcap_{b\in B} \PosEvent\ \bigg|\ \bigcap_{F\in \mathfrak F} \FailEvent\right)\\
    &\overset{\text{\Cref{eq:HarrisConditional}}}{\le} \Prob\left(\bigcap_{b\in B} \PosEvent\right) = \prod_{b\in B}\Prob\left(\PosEvent\right)\\
    & = \left(\frac 12\right)^{19\ceil*{\frac{\log_{16} n}{2}}} \le \left(\frac 12\right)^{19\frac{\log_{16} n}{2}} = \left(\frac 12\right)^{\frac {19}2\frac{\log_{\frac 12} n}{-4}} = n^{-\frac {19}8}\text.
\end{align*}

For the equality in the second line, we have used that the conditional utility of $a$ for agents that are both outside of the coalition of $a$ in $\alg(N,u)$ and not in $N^{\ind(a)}$ only depends on unsuccessful merge events.
In the second-to-last line, we use independence of the utilities for $a$ by different agents.
\end{proof}

We can apply \Cref{lem:correlation} to prove the next theorem.

\begin{restatable}{theorem}{IndStability}\label{thm:ind-stability}
Consider random hedonic games for $\distr = U(-1,1)$. Then,
\Cref{alg:main} produces an {\env} partition with probability tending to~$1$ in the large agent limit.
\end{restatable}

\begin{proof}
    We will prove \Cref{eq:highcondprob} for the case where $\mathfrak P$ represents {\envNoun}.
    Let $\partition^* = \alg(N,u)$, i.e., $\partition^*$ denotes the output of \Cref{alg:main} for $(N,u)$.
    Then, 
    \begin{align*}
        &\Prob\left(\alg(N,u) \text{ \env} \mid \Stwo\right)\\
        &= 1 - \Prob\left(\bigcup_{a\in N, C\in \partition^*\setminus\{\partition^*(a)\}} \OutEvent \ \Big|\ \Stwo\right)\\ 
        & \overset{\text{\Cref{lem:correlation}}}{\ge} 1 - n^2 n^{-\frac{19}8}\to 1 \text{ for }n \to \infty.
    \end{align*}
    There, we use a union bound and that there are at most $n$ agents in $N$ and $n$ coalitions in $\partition^*$.
\end{proof}

Recall that every individually rational and {\env}
partition is individually stable.
Hence, \Cref{thm:ir,thm:ind-stability} imply our first main result.

\begin{theorem}
    Consider random hedonic games for $\distr = U(-1,1)$. Then,
\Cref{alg:main} produces an individually stable partition with probability tending to~$1$ in the large agent limit.
\end{theorem}

\subsection{\EXVNoun}\label{sec:EXV}

Compared to {\envNoun}, it is far easier to obtain {\exvNoun}.
Agents that are part of a merged coalition after Stage~2 are denied to abandon their coalition by the agents in their clique coalition from Stage~1.
For agents from the remainder set, we show that they are denied to leave with high probability.
For this, it is useful that we implemented Stage~3 according to \Cref{spec:stagethree}.
Hence, the denial follows by considering the large set of agents outside $N^{\ind(a)}$, for which we can still assume that utilities are distributed as $U(-1,1)$.

\begin{restatable}{theorem}{THMexitdenial}
    Consider random hedonic games for $\distr = U(-1,1)$. Then,
\Cref{alg:main} produces an {\exv} partition with probability tending to~$1$ in the large agent limit.
\end{restatable}

\begin{proof}
    By \Cref{thm:cliquestage,thm:mergingstage,thm:finalstage}, \Cref{alg:main} succeeds in Stages~1,~2 and~3 with probability tending to~$1$ and we assume their success henceforth.
    By the success of Stage~3, all agents in the remainder set are assigned coalitions, for which no utilities were revealed in Stage~2.
    Let $\partition^*$ denote the produced partition of \Cref{alg:main}.

    We make a case distinction to consider when agents are denied to abandon their coalition.
    Let $\agone\in N$ and assume first that $\agone\notin R$. 
    Then, in the partition produced by \Cref{alg:main}, every agent in the clique coalition $\coa(\agone)$ is still in the coalition of $\agone$ and denies $\agone$ to leave.
    Therefore, conditioned on the success of all stages, with probability~$1$, these agents are denied to leave their final coalition.

    Now, assume that $a\in R$ and let $C:=\partition^*(a)$.
    By the success in Stage~3, no utilities of agents in $C\setminus N^{\ind(a)}$ were revealed in the execution of the algorithm.\footnote{Note, however, that in Stage~3, utilities by $a$ for agents in $C$ were revealed.}
    By success of Stage~2, it holds that $|C\setminus N^{\ind(a)}| = 19\ceil*{\frac{\log_{16} n}{2}}$. 
    
    Hence, by the principle of deferred decisions, there exists no agent in the produced coalition of agent~$a$ that denies $a$ to leave with probability at most 
    $$\left(\frac 12\right)^{19\ceil*{\frac{\log_{16} n}{2}}} \le \left(\frac 12\right)^{19\frac{\log_{16} n}{2}} = \left(\frac 12\right)^{\frac {19}2\frac{\log_{\frac 12} n}{-4}} = n^{-\frac {19}8}\text.$$

    By a union bound, since there are at most $n$ agents and therefore at most $n$ agents in the remainder set, every agent in the remainder set is denied to leave their coalition with a probability of at least $1- n\cdot n^{-\frac {19}8} = 1- n^{-\frac {11}8}$\text.
\end{proof}

This immediately implies contractual Nash stability.

\begin{theorem}
    Consider random hedonic games for $\distr = U(-1,1)$. Then,
\Cref{alg:main} produces a contractually Nash-stable partition with probability tending to~$1$ in the large agent limit.
\end{theorem}

\section{A Case against Nash stability}\label{sec:Nash}
In \Cref{prop:NashBias}, we have seen that it is easy to satisfy Nash stability if utilities are distributed with a bias to be positive: we can simply form the grand coalition, for which Nash stability is equivalent to individual rationality.
Then, for a large number of agents, the grand coalition is individually rational for every agent with high probability.
However, this result depends very much on the positive bias of the utility distribution.
In this section, we will show that Nash-stable partitions do only exist with vanishing probability in random hedonic games when utilities are sampled from $U(-1,1)$.
This is a very interesting contrast to the stability guarantees of \Cref{alg:main} and suggests that Nash stability might be too great a demand. 

Our analysis depends on a sophisticated counting argument.
We will bound the sum of the probabilities for any possible partition to be Nash-stable by a function that tends to $0$ in the large agent limit.
For this, we obtain an expression for the probability that a partition into $k$ coalitions is Nash-stable and multiply it with the $k$th Stirling number of the second kind.

The proof is structured into a series of lemmas.
We defer some of the technical details to the appendix in order to focus on the intuition and key steps of our argument. 

We start with a lemma that will be useful in bounding the probability that a partition into $k$ coalitions is Nash-stable.
Its proof relies on the use of Lagrange multipliers and can be found in the appendix.
\begin{restatable}{lemma}{bound}\label{bound}
    Let $n,k\in \mathbb{N}$.
Let $s_1,\ldots,s_{k}\in \mathbb N$ be positive integers such that $\sum_{i=1}^{k} s_i=n$.
Let $z\ge 0$ and let $q_1,\ldots,q_k\in \mathbb R$ be nonnegative real numbers such that $\sum_{i=1}^{k} q_i=z$.
Then, it holds that
$$ \prod_{i=1} q_i^{s_i}\leq \left(\frac{z}{k}\right)^{n}\text.$$
\end{restatable}

When an agent checks whether they can perform a Nash deviation, they want to decide which of the possible sums of random variables is best for them.
Hence, we would like to bound the probability of a partition with $k$ coalitions being Nash-stable with the bound in \Cref{bound}, where the $q_i$ should represent the probability that the $i$th sum of random variables is maximum and the $s_i$ represent sizes of random variables.
There is, however, a technical challenge that makes it hard to apply the \Cref{bound} directly:
When an agent considers their own coalition, they also represent a utility value and therefore, when checking whether their own coalition, say the $i$th coalition, maximizes the utility, they consider only a sum of $s_i - 1$ utility values.
The key technical challenge is to show that adding an additional random variable to this sum can only make the agent better off.
The next lemmas help with this endeavor.

For $m\in \mathbb N$ define $\distr_m$ to be the distribution of the sum of $m$ i.i.d. samples from $\distr$.

\begin{restatable}{lemma}{technical}\label{lem:technical} 
Let $\distr = U(-1,1)$. Let $Z_1$ be sampled from $\mathcal{D}_{m(1)}$ where $m(1)\in \mathbb N$ for $i\in [k]$.
Let $Y$ be a continuous random variable with cumulative density function $f_Y$ satisfying $f_Y(y)\geq f_Y(-y)$ for every $y\geq 0$.
Let $X\sim \mathcal{D}$.
Then, $\Prob(Z_1\geq Y)\leq \Prob(Z_1+X\geq Y)$.
\end{restatable}

Unfortunately, \Cref{lem:technical} is not sufficient for dealing with singleton coalitions because in this case, we want to add a random variable for the empty coalition represented by a utility value of~$0$.
We deal with this case by proving a variation of \Cref{lem:technical} suitable for 
singleton coalitions.
For its proof, we make use of yet another lemma that deals with the distribution of the maximum of several random variables.

\begin{restatable}{lemma}{LemMax}\label{lem:max}
Let $Z_1,\ldots, Z_k$ be independent symmetric continuous random variables with mean $0$ 
and let $Y=\max(Z_1,\ldots,Z_k)$.
Then for any $x\geq 0$, it holds that $$f_Y(x)\geq f_Y(-x)$$ where $f_Y$ is the probability density function of $Y$.
\end{restatable}

This leads to the version of \Cref{lem:technical} applicable to 
singleton coalitions.

\begin{restatable}{lemma}{technicalzero}\label{lem:technicalzero} 
Let $\distr = U(-1,1)$. 
Let $Z_2,\ldots,Z_k$ be independent random variables where $Z_i$ is sampled from $\mathcal{D}_{m(i)}$ where $m(i)\in \mathbb N$ for $i\in [k]$.
Let $X\sim \mathcal{D}$.
Then, $\Prob(0\geq \max(Z_2,\ldots,Z_k))\leq \Prob(X\geq \max(Z_2\ldots,Z_k))$.
\end{restatable}

We can apply the insight from the previous lemmas to bound the probability that a fixed partition into $k$ coalitions is Nash-stable. 
We apply the proof strategy discussed after \Cref{bound} and apply \Cref{lem:technical,lem:technicalzero} to add an additional random variable.

\begin{restatable}{lemma}{allbound}\label{lem:allbound}
Let $\distr = U(-1,1)$ and consider a random hedonic game $(N,u)\sim \randgame$.
Let $\partition$ be a partition of $N$ into $k$ coalitions.
Then, $$\Prob\left(\partition \text{ is Nash-stable for }(N,u)\right)\le \left(\frac{1}{k}\right)^n\text.$$
\end{restatable}

\begin{proof}
For $\distr = U(-1,1)$ let $(N,u)\sim \randgame$ and consider a partition $\partition = \{C_1,\dots, C_k\}$ into $k$ coalitions.

For $a\in N$, define $p_a :=\Prob\left(u_a(\partition) = \max\{u_a(C_i\cup\{a\})\colon i \in [k]\}\right)$, i.e., the probability that $a$ cannot perform a Nash-deviation to join any other nonempty coalition.
By independence of the utilities, it follows that 
\begin{equation}\label{eq:NashReal}
    \Prob\left(\pi \text{ is Nash-stable for }(N,u)\right) \le \prod_{a\in N}p_a\text.
\end{equation}
Note that this is an inequality because we disregard the deviations to form a singleton coalition.

Now, for $i\in [k]$ and $j\in [|C_i|]$, let $X_{i,jt}\sim U(-1,1)$ and $Z_i=\sum_{j=1}^{|C_i|}X_{i,j}$.
Let $q_i:=\Prob(Z_i\geq \max(Z_1,\ldots,Z_n))$.
Note that $Z_i$ is the sum of $|C_i|$ samples from random variables distributed according to $U(-1,1)$.
Hence, if $a\notin C_i$, this is precisely the distribution of $u_a(C_i\cup\{a\})$.
Moreover, if $a\in C_{i^*}$, then $u_a(\partition)$ has the distribution of the sum of $|C_{i^*}|-1\geq 0$ random variables distributed according to $U(-1,1)$.

Let $Y=\max(Z_2,\ldots,Z_k)$. Then since $|C_i|\geq 1$ for each $i=2,\ldots,k$, $Y$ has a probability density function $f_Y$. 
By \Cref{lem:max}, it holds that  $f_Y(y)\geq f_Y(-y)$ whenever $y\geq 0$.
If $|C_{i^*}|\ge 2$, we can apply \Cref{lem:technical} to infer that $p_a\leq q_{i^*}$.
If $|C_{i^*}|= 1$, it follows from \Cref{lem:technicalzero} that $p_a\leq q_{i^*}$.

Combining this with \Cref{eq:NashReal}, we obtain 
$$
\Prob\left(\pi \text{ is Nash-stable for }(N,u)\right)\le \prod_{a\in N}p_a\leq \prod_{i=1}^{k}q_i^{|C_i|}.$$
Since $\sum_{i=1}^k q_i=1$, using \Cref{bound},
we conclude that $$\Prob\left(\pi \text{ is Nash-stable for }(N,u)\right)\leq \left(\frac{1}{k}\right)^n\text.$$
\end{proof}

We can improve this bound for the case where the partition does not contain singleton coalitions.
\begin{restatable}{lemma}{nosinglebound}\label{lem:nosinglebound}
Let $\distr = U(-1,1)$ and consider a random hedonic game $(N,u)\sim \randgame$.
Let $\partition$ be a partition of $N$ into $k$ nonsingleton coalitions.
Then, $$\Prob\left(\partition \text{ is Nash-stable for }(N,u)\right)\le\left(\frac{1-\frac12^k}{k}\right)^n.$$
\end{restatable}

\begin{proof}
For $\distr = U(-1,1)$ let $(N,u)\sim \randgame$ and consider a partition $\partition = \{C_1,\dots, C_k\}$ into $k$ coalitions such that $|C_i|\ge 2$ for all $i\in [k]$.
First note that if $k = 1$, then $\partition$ is the grand coalition, which is Nash-stable with probability $\left(\frac 12\right)^n$, and the bound is correct.
We may therefore assume that $k\ge 2$.

Similar to the proof of \Cref{lem:allbound}, for $a\in N$ define $$p_a := \Prob\left(u_a(\partition)\geq\max\{u_a(C_1\cup\{a\}),u_a(C_2\cup\{a\}),\ldots u_a(C_k\cup\{a\}),0\}\right).$$
The difference is that we now also want a utility larger than $0$, i.e., to be better than the singleton coalition.
Hence, 
\begin{equation}\label{eq:exactproduct}
    \Prob\left(\partition \text{ is Nash-stable for }(N,u)\right) = \prod_{a\in N}p_a\text.
\end{equation}

Moreover, for $i\in [k]$ and $j\in [|C_i|]$, let again $X_{i,j}\sim U(-1,1)$ and $Z_i=\sum_{j=1}^{|C_i|}X_{i,j}$.
Set $Z_{k+1} = 0$.
However, we also adjust the definition of $q_i$ to include~$0$ and define $$q_i:=\Prob(Z_i\geq \max(Z_1,\ldots,Z_k,Z_{k+1}))$$ for $i\in [k+1].$
Again, $Z_i$ is the sum of $|C_i|$ random variables sampled from $U(-1,1)$.
Hence if $a\notin C_i$, then $u(C_i\cup\{a\})$ is distributed identically to $Z_i$.
Moreover, if $a\in C_{i^*}$, then and $u_a(C_{i^*})$ has the distribution of the sum of $|C_{i^*}|-1$ random variables distributed according to $U(-1,1)$.
Note that $\partition$ does not contain singleton coalitions and therefore $|C_{i^*}|-1\ge 1$.
Moreover, for $i\in [n]$, it holds that $Z_i$ is the sum of at least one random variable.
By the law of total probability, we have that $$q_i=\Prob\left(Z_i\geq 0\text{ and }0\geq \max_{j\in[k],j\neq i}\{Z_j\}\right)+\Prob\left(Z_i\geq \max_{j\in[k],j\neq i}\{Z_j\}\text{ and } \max_{j\in[k],j\neq i}\{Z_j\}>0\right).$$
Similarly if $a\in C_i$, then $$p_a=\Prob\left(u_a(C_i)\geq 0\text{ and }0\geq \max_{j\in[k],j\neq i}\{Z_j\}\right)+\Prob\left(u_a(C_i)\geq \max_{j\in[k],j\neq i}\{Z_j\}\text{ and } \max_{j\in[k],j\neq i}\{Z_j\}>0\right).$$
The first term in both of these expressions equals  $\left(\frac{1}{2}\right)^n$, since all $Z_j\leq 0$ for every $j\neq i$ w.p. $\left(\frac{1}{2}\right)^{n-1}$ and then $Z_i$ (or $u_a(C_i)$ respectively)
are positive with probability $\frac{1}{2}$ since it is the sum of $|C_i|\geq 2$ (respectively, $|C_i|-1\geq 1$) $U(-1,1)$ random variables.
To conclude that $p_a\leq q_{i^*}$
we will apply \Cref{lem:technical} to show that 
$$\Prob\left(u_a(C_i)\geq \max_{j\in[k],j\neq i}\{Z_j\}\ \Big|\ \max_{j\in[k],j\neq i}\{Z_j\}>0\right)\leq \Prob\left(Z_i\geq \max_{j\in[k],j\neq i}\{Z_j\}\ \Big|\  \max_{j\in[k],j\neq i}\{Z_j\}>0\right)\text.$$
This is sufficient since by conditional probability, 
\begin{align*}
    &\Prob\left(u_a(C_i)\geq \max_{j\in[k],j\neq i}\{Z_j\}\text{ and } \max_{j\in[k],j\neq i}\{Z_j\}>0\right)\\
    &=\Prob\left(u_a(C_i)\geq \max_{j\in[k],j\neq i}\{Z_j\}\ \Big|\ \max_{j\in[k],j\neq i}\{Z_j\}>0\right)\cdot \Prob\left(\max_{j\in[k],j\neq i}\{Z_j\}>0\right)
\end{align*}
 and an analogous statement holds for $Z_i$ in place of $u_a(C_i)$.
 
Let $Y=\left\{\max_{j\in[k],j\neq i}\{Z_j\}>0)\ \big|\ \max_{j\in[k],j\neq i}\{Z_j\}>0\right\}$.
Since $u_a(C_i)$ and $Z_i$ are both independent of $Y$,
$$\Prob\left(u_a(C_i)\geq \max_{j\in[k],j\neq i}\{Z_j\}\ \Big|\ \max_{j\in[k],j\neq i}\{Z_j\}>0\right)=\Prob(u_a(C_i)\geq Y)$$
and $$\Prob\left(Z_i\geq \max_{j\in[k],j\neq i}\{Z_j\}\ \Big|\ \max_{j\in[k],j\neq i}\{Z_j\}>0\right)=\Prob(Z_i\geq Y)$$
Observe that $Y$ is a continuous random variable and its cumulative density function $f_Y$ trivially satisfies $f_Y(y)\geq f_Y(-y)$ for any $y\geq 0$ since $f_Y(y)=0$ whenever $y<0$.
Since the conditions of \Cref{lem:technical} are satisfied, it follows that $\Prob(u_a(C_i)\geq Y)\leq\Prob(Z_i\geq Y)$, showing that $p_a\leq q_{i^*}$, as desired.
Hence, by \Cref{eq:exactproduct}, we obtain 
$$
\Prob\left(\partition \text{ is Nash Stable for }(N,u)\right)=\prod_{a\in N}p_a\leq \prod_{i=1}^{k}q_i^{|C_i|}\text.$$
Now the probability that $0$ is the maximum among $Z_1,\ldots,Z_k, Z_{k+1}$ is exactly $q_{k+1} = \frac{1}{2}^k$.
Indeed, it is maximum if and only if all of the $Z_i$ are negative. 
It follows that $\sum_{i=1}^kq_i=1-\left(\frac{1}{2}\right)^k$.
Using \Cref{bound},
we conclude that 
$$\Prob\left(\partition \text{ is Nash Stable for }(N,u)\right)\leq \left(\frac{1-\left(\frac{1}{2}\right)^k}{k}\right)^n.$$
This completes the proof.
\end{proof}

We can combine the previous lemmas to prove the main theorem of this section.
As we discussed in the beginning of the section, we essentially want to multiply the probabilities obtained in \Cref{lem:allbound,lem:nosinglebound} with the Stirling number of the second kind.
Our analysis relies on two estimations that consider the cases of a small and a large number of coalitions.

\begin{restatable}{theorem}{nashvanish}\label{thm:nashvanish}
Let $\distr = U(-1,1)$ and consider a random hedonic game $(N,u)\sim \randgame$.
Then, 
	\begin{align*}
		\lim_{n\to \infty} {\Prob}((N,u) \text{ contains Nash-stable partition}) = 0.
	\end{align*}
\end{restatable}

\begin{proof}
Our goal is to give a bound for the probability that there exists a Nash-stable partition into $k$ coalitions.

Therefore, we want to bound the product of the bounds on the probability that such a partition is Nash-stable obtained in \Cref{lem:allbound,lem:nosinglebound} with the number of such partitions. 
The number of partitions of $n$ elements into $k$ sets is given by the Stirling number of the second kind \begin{align*}\stirling{n}{k}=\sum_{i=1}^k(-1)^{k-i}\frac{i^n}{i!(k-i)!},\end{align*}

We make a case distinction based on the size of $k$.
Assume first that $k>\frac{\log_{2}n}{2}.$ 
By \Cref{lem:allbound}, for any of the $\stirling{n}{k}$ partitions into k coalitions, the probability that one is stable is bounded by $\left(\frac{1}{k}\right)^n$. 
We show that the probability that any partition of size $k$ is Nash-stable is bounded by $2^{-\Omega(\log n \log\log n)}$.
Therefore, observe that
    \begin{align}
    &\stirling{n}{k}\left(\frac{1}{k}\right)^n=\left(\frac{1}{k}\right)^n\sum_{i=1}^k(-1)^{k-i}\frac{i^n}{i!(k-i)!}\leq \sum_{i=1}^k\frac{i^n}{k^ni!(k-i)!}\notag\\
    &\leq \sum_{i=1}^k\frac{1}{i!(k-i)!}\leq k\cdot \frac{1}{\ceil*{\frac{k}{2}}!\floor*{\frac{k}{2}}!}\le \frac{2}{\ceil*{\frac{k}{2}-1}!\floor*{\frac{k}{2}}!}\label{centralbinom}\\
    & \leq  \frac{1}{\ceil*{\frac{\log_{2}n}{4}-1}!\floor*{\frac{\log_{2}n}{4}}!} <  \frac{1}{\left(\ceil*{\frac{\log_{2}n}{5}}!\right)^2}\label{plugin}\\
    &
    <\frac{1}{2\pi\ceil*{\frac{\log_{2}n}{5}}}\left(\frac{e}{\ceil*{\frac{\log_{2}n}{5}}}\right)^{2\cdot \ceil*{\frac{\log_{2}n}{5}}}\cdot e^{-\frac{2}{12\cdot \ceil*{\frac{\log_{2}n}{5}}+1}}
    \label{stirlapprox}\\
    &\leq\left(\frac{e}{\frac{\log_{2}n}{5}}\right)^{2\cdot (\frac{\log_{2}n}{5})}\cdot\frac{1}{2\pi\frac{\log_{2}n}{5}} e^{-\frac{2}{12\cdot \ceil*{\frac{\log_{2}n}{5}}+1}}
    \notag\\
    &=\log_2 n^{-\frac{2}{5}\log_2 n}\cdot\frac{(5e)^{\frac{2}{5} \log_{2}n}}{2\pi\frac{\log_{2}n}{5}} e^{-\frac{2}{12\cdot \ceil*{\frac{\log_{2}n}{5}}+1}}
    \notag\\
    &=2^{-\frac{2}{5}\log_2 n \log_2 \log_2 n}\cdot\frac{(5e)^{\frac{2}{5} \log_{2}n}}{2\pi\frac{\log_{2}n}{5}} e^{-\frac{2}{12\cdot \ceil*{\frac{\log_{2}n}{5}}+1}}
    \notag\\
    &=2^{-\Omega(\log n \log\log n)}\notag
    \end{align}
Inequality (\ref{centralbinom}) holds since the central binomial coefficient is largest. The first inequality in (\ref{plugin}) holds using $k>\frac{\log_{2}n}{2}$ and the second holds for large enough $n$. Inequality (\ref{stirlapprox}) is an application of the Stirling inequality $\ell!>\sqrt{2\pi \ell}\left(\frac{\ell}{e}\right)^{\ell} e^{\frac{1}{12\ell+1}}$, which holds for every $\ell\in \mathbb{N}$.

Now consider $k\leq\frac{\log_{2}n}{2}$ and suppose $\ell$ of the $k$ coalitions are singletons.
We first restrict attention to the partial partition consisting of the coalitions of size at least~$2$.
By \Cref{lem:nosinglebound}, the probability that any partition of $n-\ell$ agents into $k-\ell$ coalitions, none of which are singletons, is Nash-stable can be bounded as follows.
  \begin{align}
  &\stirling{n-\ell}{k-\ell}\left(\frac{1-\left(\frac{1}{2}\right)^{k-\ell}}{k-\ell}\right)^{n-\ell}\notag\\
  &=\left(\frac{1-\left(\frac{1}{2}\right)^{k-\ell})}{k-\ell}\right)^{n-\ell}\sum_{i=1}^{k-\ell}(-1)^{k-\ell-i}\frac{i^{n-\ell}}{i!(k-\ell-i)!}\notag\\ 
  &\leq \left(1-\left(\frac{1}{2}\right)^{k-\ell}\right)^{n-\ell}\sum_{i=1}^{k-\ell}\frac{i^{n-\ell}}{(k-\ell)^{n-\ell}i!(k-\ell-i)!}\notag\\
  &<\left(1-\left(\frac{1}{2}\right)^{k-\ell}\right)^{n-\ell}\sum_{i=1}^{k-\ell}\frac{1}{i!(k-\ell-i)!}\notag\\
  &\leq e\cdot \left(1-\left(\frac{1}{2}\right)^{k-\ell}\right)^{n-\ell}\label{ebound}\\
  &\leq e\cdot\left(1-\frac{1}{2}^{\frac{\log_{2}n}{2}}\right)^{n-\ell}=e\cdot\left(1-\frac{1}{\sqrt{n}}\right)^{n-\ell}\notag\\
  &\le e\cdot\left(1-\frac{1}{\sqrt{n}}\right)^{n-\log_2 n}\leq e^{1-\sqrt{n}+\frac{\log_2 n}{\sqrt{n}}}\notag\end{align}
Here \Cref{ebound} follows since $\sum_{i=1}^k\frac{1}{i!(k-i)!}\leq \sum_{i=1}^k\frac{1}{i!}< \sum_{i=1}^\infty\frac{1}{i!}=e$.\\
There are $\binom{n}{\ell}$ ways of choosing the $\ell$ singletons and the probability that none of the $\ell$ singletons can perform a Nash deviation is $\left(\frac{1}{2}\right)^{\ell(k-1)}$.
The probability that a partition into $k$ coalitions is Nash-stable is therefore bounded by
\begin{align*}\stirling{n-\ell}{k-\ell}\cdot \binom{n}{\ell }\left(\frac{1}{2}\right)^{\ell(k-1)}\cdot\left(\frac{1-\left(\frac{1}{2}\right)^{k-\ell}}{k-\ell}\right)^{n-\ell}\leq n^{\ell}\cdot e^{1-\sqrt{n}+\frac{\log_2 n}{\sqrt{n}}}\leq n^{\log_2 n}e^{1-\sqrt{n}+\frac{\log_2 n}{\sqrt{n}}}
= e^{\Omega(-\sqrt{n})}\end{align*}

We put everything together and obtain
\begin{align*}
 &\lim_{n\rightarrow \infty }\Prob((N,u)\text{ contains Nash-stable partition})\\
 &=
 \lim_{n\rightarrow \infty }\sum_{k=1}^{n} \Prob((N,u)\text{ contains Nash-stable partition in }k \text{ coalitions}) \\
 &\leq \lim_{n\rightarrow \infty }\sum_{k=1}^{n} \max\left(2^{-\Omega(\log n \log\log n)},e^{\Omega(-\sqrt{n})}\right)\leq \lim_{n\rightarrow \infty}n\max\left(2^{-\Omega(\log n \log\log n)},e^{\Omega(-\sqrt{n})}\right)=0,
 \end{align*}
The first equality follows from partitioning the event of Nash stability into disjoint events. 
We then apply the maximum of the two bounds for different sizes of $k$.
This concludes the proof.

\end{proof}

\section{Conclusion}

We proposed a model of random hedonic games that offers a novel perspective on the analysis of hedonic games whose study has thus far been restricted to a deterministic setting.
Our model describes a random variant of additively separable hedonic games, where utilities are drawn i.i.d. from a given distribution.
While trivial coalition structures like the grand coalition already guarantee Nash stability and contractual Nash stability under certain conditions,
our work focuses on the arguably difficult case of \textit{impartial culture} where utilities are sampled from the uniform distribution $U(-1,1)$.
Our main result is an algorithm that with high probability determines a partition into $\Theta(\frac{n}{\log n})$ coalitions of size $\Theta(\log n)$ which are individually rational, {\env}, and {\exv}.
This implies in particular that the produced partitions are individually stable and contractually Nash-stable.
By contrast, with high probability, Nash-stable partitions do not exist in random hedonic games for $U(-1,1)$-distributed utilities.

Individual stability and contractual Nash stability are often seen as complementary solution concepts.
Nevertheless, there are crucial conceptual differences and our results contribute to understanding these differences.
First, contractual deviations treat joining a nonempty coalition or forming a new coalition similar because {\exvNoun} applies to both. 
By contrast, {\envNoun} does not constrain creating a new coalition.
Hence, individual rationality as implied by individual stability appears more closely related to the stronger concept of Nash stability.
Moreover, satisfying {\exvNoun} appears easier to satisfy than {\envNoun}.
Essentially, denying an agent to leave prevents \emph{every} contractual deviation by this agent whereas denying an agent to join only blocks \emph{some} individual deviation.
In other words, to satisfy {\exvNoun}, we have to fulfill one constraint for every agent, whereas for {\envNoun}, we have to satisfy $|\partition|-1$ constraints per agent.
Our results reflect this intuition because {\exvNoun} is easily satisfied in the probabilistic setting whereas {\envNoun} (as well as individual rationality) are more difficult to achieve.

This work is a stepping stone for a range of fascinating future directions.
First, it would be interesting to consider random models for other classes of hedonic games.
Interestingly, our results already have implications for other classes of hedonic games based on cardinal valuation functions of single agents.
We show that, with high probability, \Cref{alg:main} produces partitions such that $(i)$ for every agent, the sum of utilities for agents in their own coalition is positive and $(ii)$ for every agent there exists an agent in every other coalition that has a negative utility for this agent.
Property~$(i)$ is equivalent to individual rationality not only for additively separable hedonic games but also for other commonly considered games, such as fractional and modified fractional hedonic games \citep{ABB+17a,Olse12a}.
Moreover, while in these classes of games property~$(ii)$ is not sufficient to guarantee {\envNoun}, it is sufficient under individual rationality.
Hence, \Cref{alg:main} even finds individually rational and {\env}, and therefore individually stable, partitions in random fractional and modified fractional hedonic games.

By contrast, \Cref{alg:main} guarantees {\exvNoun} for additively separable hedonic games by guaranteeing that for every agent some agent in their own coalition has a positive value for them.
This condition is not equivalent to {\exvNoun} for fractional and modified fractional hedonic games: Example~2.2 by \citet{BBK23a} shows that it may be preferable for an agent to be abandoned by another agent for which their utility is positive.
We suspect that a stronger version of {\exvNoun} holds for the output of \Cref{alg:main}, which would yield contractual Nash stability for fractional and modified fractional hedonic games.
Of course, the landscape of classes is far richer than the discussed classes of games and one could also define random variants of other classes of games as well as a richer class of utility distributions.

\section*{Acknowledgements}
Martin Bullinger was supported by the AI Programme of The Alan Turing Institute.
Sonja Kraiczy was supported by an EPSRC studentship.
A preliminary version of this article appears in the 
Proceedings of the 25th ACM Conference on Economics and Computation (2024).
We would like to thank the anonymous reviewers for their helpful feedback.

\appendix

\section*{Appendix: Omitted Proofs from Section \ref{sec:Nash}}

In this appendix, we present the missing proofs of the lemmas leading to the nonexistence of Nash-stable partitions with high probability for games with utilities sampled from $U(-1,1)$. 
We start with the first central lemma.

\bound*

\begin{proof}
Assume that all numbers are defined as in the statement of the lemma.
We will prove the bound via two consecutive applications of the method of Lagrange multipliers.
First we want to maximize $\prod_{j=1} q_i^{s_i}$ subject to the constraint $\sum_{i=1}^{k} q_i=z$. Suppose for each $i=1,\ldots,k$, it holds that $q_i\neq 0$, otherwise the product is trivially $0$ since by assumption each $s_i$ is positive.
We consider $q_1,q_2,\ldots, q_k$ to be variables and the $s_i$ are fixed.
By applying the natural logarithm to our objective function, this is equivalent to the following maximization problem:

\begin{equation}\label{eq:pisum}
    \text{maximize }\sum_{i=1}^k s_i \ln(q_i) \text{ subject to } \sum_{i=1}^{k} q_i=z
\end{equation}
The Lagrangian function for \Cref{eq:pisum} is given by
$$\sum_{i=1}^k s_i \ln(q_i)-\lambda \left(\sum_{i=1}^{k} q_i-z\right)\text.$$
By the KKT conditions, setting the partial derivative with respect to $q_i$ to $0$
gives $$\frac{s_i}{q_i}=\lambda \iff q_i=\frac{s_i}{\lambda}\text.$$
Summing over all $i$ we obtain $$z=\sum_{i=1}^k q_i=\frac{1}{\lambda}\sum_{i=1}^k s_i=\frac{n}{\lambda},$$
which implies that $$\lambda=\frac{n}{z}$$ and hence $$q_i=\frac{s_i}{n}z.$$
Hence, we know that 
\begin{equation}\label{eq:firstbound}
    \prod_{j=1}^k q_i^{s_i} \leq \prod_{j=1}^k \left(\frac{s_i}{n}z\right)^{s_i}\text.
\end{equation}
We now use Lagrange multipliers again for maximizing the right hand side of \Cref{eq:firstbound}.
Taking the natural logarithm of the objective function, this yields the maximization problem
\begin{equation}\label{eq:sisum}
    \text{maximize }%
\sum_{i=1}^k s_i \ln\left(\frac{s_i}{n}z\right)\text{ subject to }\sum_{i=1}^{k} s_i=n\text.
\end{equation}
The Lagrangian function for this maximization problem is given by
$$\sum_{i=1}^k s_i \ln\left(\frac{s_i}{n}z\right)-\lambda \left(\sum_{i=1}^{k} s_i-n\right)\text,$$
and setting the partial derivative with respect to $s_i$ to $0$
gives $$\ln\left(\frac{s_i}{n}z\right)+1-\lambda =0\iff s_i=\frac{e^{\lambda-1}}{z}n.$$
Note that $s_i=s_j$ for all $i,j\in[k]$. 
Hence, the maximum of \Cref{eq:sisum} is achieved when all $s_i$ are as large as possible subject to $k s_i = n$.
This happens when $s_i=\frac{n}{k}$ for all $i$.
Combining this with \Cref{eq:firstbound}, we obtain
$$\prod_{j=1}^k q_i^{s_i}\leq \prod_{j=1}^k \left(\frac{s_i}{n}z\right)^{s_i}\leq \prod_{j=1}^k \left(\frac{z}{k}\right)^{\frac{n}{k}}=\left(\frac{z}{k}\right)^n,$$

as required.
\end{proof}

Towards the proof of the next lemma, we prove another auxiliary lemma.
    It captures how the distribution of the sum of $U(-1,1)$-distributed random variables changes when adding another term. 
    
\begin{lemma}\label{lem:var}
Let $\mathcal{D}=U(-1,1)$ and let $Z_1,\ldots,Z_{k+1}$, $k\geq1$, be random variables distributed i.i.d. according to~$\mathcal{D}$.
    Then, for any $x\geq 0$ it holds that
    $$\Prob\left(\sum_{i=1}^{k+1}Z_i\geq x\right)\geq \Prob\left(\sum_{i=1}^{k}Z_i\geq x\right).$$
\end{lemma}
\begin{proof}
    Let $g$ be the probability density function of $\sum_{i=1}^{k}Z_i$ with cumulative density function $G$, let $F$ be the cumulative density function of $\sum_{i=1}^{k+1}Z_i$, and let $H$ be the cumulative density function of $Z_{n+1}$.
    Using independence for the first step, we obtain
    \begin{align*}
    1-F(x)&=\int_{-\infty}^{\infty} g(y)(1-H(x-y))dy\\
    &=\int_{-\infty}^xg(y)(1-H(x-y))dy+\int_{x}^{\infty} g(y)(1-H(x-y))dy\\
    &=\int_{-\infty}^xg(y)(1-H(x-y))dy+\int_{x}^{\infty} g(y)dy-\int_{x}^{\infty} g(y)H(x-y))dy\\
    &=\int_{-\infty}^xg(y)(1-H(x-y))dy+ 1-G(x) -\int_{x}^{\infty} g(y)H(x-y)dy
    \end{align*}
To show that $1-F(x)\geq 1-G(x)$, we will show that $$\int_{-\infty}^xg(y)(1-H(x-y))dy-\int_{x}^{\infty} g(y)H(x-y)dy\geq 0.$$

    To see this, we compute
    \begin{align*}
    &\int_{-\infty}^xg(y)(1-H(x-y))dy-\int_{x}^{\infty} g(y)H(x-y)dy\\&=\int_{0}^{\infty}g(x-y) (1-H(y))dy-\int_{0}^{\infty}g(z+x)H(-y)dy
    \\
    &=\int_{0}^{\infty}g(x-y) (H(-y))dy-\int_{0}^{\infty}g(z+x)H(-y)dy
    \\
    &=\int_{0}^{\infty}\left(g(x-y)-g(x+z)\right) H(-y)dy\geq 0\text.
    \end{align*}
    For the second equality, we use that $H$ is the cumulative density function of a symmetric random variable.
    The last inequality follows because $g$ is unimodal and symmetric and therefore $g(x-y)\geq g(x+y)$ for any $x,y\geq 0$.%
\end{proof}

We are ready to prove our next main lemma.

\technical*
\begin{proof}
Let everything be defined as in the statement of the lemma.
    
    Now let $G$ be the cumulative density function of $Z_1+X$ and let $F$ be the cumulative density function of $Z_1$. By \Cref{lem:var}, for any $y\geq 0$, it holds that $1-G(y)=\Prob(Z_1+X\geq y)\geq \Prob(Z_1\geq y)=1-F(y)$ and therefore $F(y)\ge G(y)$.
    We obtain 
    \begin{align}&\Prob(Z_1+X\geq Y)\notag\\
    &=\int_{-\infty}^{\infty}f_Y(y) (1-G(y)) dy \notag\\
    &= \int_{-\infty}^{\infty}f_Y(y) (F(y)-G(y)+(1-F(y)) dy\notag\\
    &=\int_{-\infty}^{\infty}f_Y(y) (1-F(y))+\int_{-\infty}^{\infty}f_Y(y)(F(y)-G(y))dy\notag\\
    &=\Prob(Z_1\geq Y)+\int_{-\infty}^{\infty}f_Y(y)(F(y)-G(y))dy\notag\\
    &=\Prob(Z_1\geq Y)+\int_{0}^{\infty}f_Y(y)(F(y)-G(y))dy+\int_{-\infty}^{0}f_Y(y)(F(y)-G(y))dy\notag\\
    &=\Prob(Z_1\geq Y)+\int_{0}^{\infty}f_Y(y)(F(y)-G(y))dy+\int_{0}^{\infty}f_Y(-y)(F(-y)-G(-y))dy\notag\\
    &=\Prob(Z_1\geq Y)+\int_{0}^{\infty}f_Y(y)(F(y)-G(y))dy+\int_{0}^{\infty}f_Y(-y)(1-F(y)-(1-G(y))dy \label{symmetry}\\
    &=\Prob(Z_1\geq Y)+\int_{0}^{\infty}(f_Y(y)-f_Y(-y))(F(y)-G(y))dy+0
    \notag\\
    &\geq \Prob(Z_1\geq Y)\label{last}\end{align}
    where \Cref{symmetry} follows since $Z_1+X$ and $Z_1$ are both symmetrically distributions with mean $0$ and \Cref{last} follows since for $y\geq 0$, $F(y)\geq G(y)$ and $f_Y(y)\geq f_Y(-y)$.
\end{proof}

The next lemma deals with the probability density function of the maximum of identically distributed random variables.

\LemMax*

\begin{proof}
Assume that we consider objects satisfying the assumptions of the lemma.
Let $F_Y$ be the cumulative density function of $Y$, and for $i\in [k]$, let $f_{Z_i}$ and $F_{Z_i}$ be the probability density function and cumulative density function of $Z_i$, respectively.
By independence, for any $x\in \mathbb R$, we have that 
    $$f_Y(x)=\sum_{i=1}^{k}f_{Z_i}(x)\prod_{j\neq i} F_{Z_j}(x).$$
    Hence, for $y\ge 0$, it holds that
    \begin{align}f_Y(x)&=\sum_{i=1}^{k}f_{Z_i}(x)\prod_{j\neq i} F_{Z_j}(x)\notag\\&=\sum_{i=1}^{k}f_{Z_i}(-x)\prod_{j\neq i} F_{Z_j}(x)\label{eq:dist:symmetry:one}\\&\geq \sum_{i=1}^{k}f_{Z_i}(-x)\prod_{j\neq i} F_{Z_j}(-x)\label{eq:cdf}\\&=f_Y(-x)\text,\notag\end{align}
    where \Cref{eq:dist:symmetry:one} follows since the distribution of each $Z_i$ is symmetric, and \Cref{eq:cdf} follows since $F$ is a cumulative density function and therefore $F(x)\geq F(-x)$ for any $x\geq 0$.
\end{proof}

We apply the lemma for the version of \Cref{lem:technical} applicable to 
singleton coalitions.

\technicalzero*
\begin{proof}
Let $f_Y$ be the probability density function of $Y=\max\{Z_2,\ldots,Z_k\}$ .
Then,
\begin{align*}\Prob\left(0\geq \max(Z_2,\ldots,Z_k)\right)=\Prob(0\geq Y)=\int_{-\infty}^0 f_Y(y)dy=\int_{-\infty}^{-1} f_Y(y)dy+\int_{-1}^0 f_Y(y) dy\text.\end{align*}
Furthermore, 
\begin{align*}\Prob(X\geq \max\left(Z_2,\ldots,Z_k\right))&=\Prob(X\geq Y)\\&=\int_{-\infty}^{1}\Prob(X\geq y)f_Y(y) dy=\int_{-\infty}^{-1}f_Y(y) dy+\int_{-1}^{1}\Prob(X\geq y)f_Y(y) dy.\end{align*}
To prove the claim, we show that $\int_{-1}^{1}\Prob(X\geq y)f_Y(y) dy\geq \int_{-1}^0 f_y(y) dy.$
\begin{align}\int_{-1}^{1}\Prob(X\geq y)f_Y(y) dy&=\int_{-1}^{0}\Prob(X\geq y)f_Y(y) dy+\int_{0}^{1}\Prob(X\geq y)f_Y(y) dy \notag\\&=\int_{0}^{1}\Prob(X\geq -y)f_Y(-y) dy+\int_{0}^{1}\Prob(X\geq y)f_Y(y) dy \notag\\&\geq \int_{0}^{1}\Prob(X\geq -y)f_Y(-y) dy+\int_{0}^{1}\Prob(X\geq y)f_Y(-y) dy \label{maxdensity}
\\&=\int_{0}^{1}(\Prob(X\geq -y)+P(X\geq y))f_Y(-y) dy \notag\\&=\int_{0}^{1} f_Y(-y)dy=\int_{-1}^{0} f_Y(y)dy\label{laststep},
\end{align}

where Inequality (\ref{maxdensity}) follows by \Cref{lem:max} and \Cref{laststep} follows since $X$ is symmetric and so $\Prob(X\geq -y)+\Prob(X\geq y)=\Prob(X\leq y)+\Prob(X\geq y)=1$
\end{proof}

\end{document}